\documentclass{llncs}

\usepackage{makeidx}  % allows for indexgeneration
\usepackage{times}
\usepackage{xspace}
\usepackage{amsmath, amssymb, latexsym}
\usepackage{bbm}
\usepackage{stmaryrd}
\usepackage{color}
\usepackage{dsfont}
\usepackage{microtype}%if unwanted, comment out or use option "draft"
\usepackage{boxedminipage}
\usepackage{smallsec}
\usepackage{fullpage}

\def\squarebox#1{\hbox to #1{\hfill\vbox to #1{\vfill}}}
\renewcommand{\qed}{\hspace*{\fill}
       \vbox{\hrule\hbox{\vrule\squarebox{.667em}\vrule}\hrule}\smallskip}
\renewenvironment{proof}{\begin{trivlist}
\item[\hspace{\labelsep}{\bf\noindent Proof: }]
}{\qed\end{trivlist}}

\newcommand{\shaull}[1]{}
\newcommand{\short}[1]{}
\newcommand{\set}[1]{{\left\{#1\right\}}}
\newcommand{\Nat}{\mathbbm{N}}
\newcommand{\Rat}{\mathbbm{Q}}

\newcommand{\tuple}[1]{\langle #1  \rangle}
\newcommand{\zug}[1]{\langle #1  \rangle}
\newcommand{\pair}{\tuple}
\newcommand{\sem}[1]{[\![#1]\!]}
\mathchardef\mhyphen="2D
\newcommand{\True}{\mathtt{True}}
\newcommand{\False}{\mathtt{False}}
\newcommand{\LTL}{{\ensuremath{\rm LTL}}\xspace}
\newcommand{\PLTL}{{\ensuremath{\rm PLTL}}\xspace}
\newcommand{\CTL}{{\rm CTL}}

\newcommand{\Next}{\mathsf{X}}
\newcommand{\Yest}{\mathsf{Y}}
\newcommand{\Ev}{\mathsf{F}}
\newcommand{\Alw}{\mathsf{G}}
\newcommand{\Until}{\mathsf{U}}
\newcommand{\Since}{\mathsf{S}}

\newcommand{\Release}{\mathsf{R}}

\newcommand{\parUntil}{\mathsf{O}}

\newcommand{\K}{{\mathcal K}}
\newcommand{\T}{{\mathcal T}}

\newcommand{\NGBW}{\mbox{\rm NGBA}\xspace}

\newcommand{\NBW}{\mbox{\rm NBA}\xspace}

\newcommand{\AWW}{\mbox{\rm AWA}\xspace}

\newcommand{\twAWW}{\mbox{\rm 2AWA}\xspace}

\newcommand{\gap}{\vspace*{0.25 cm}}
\newcommand{\stam}[1]{}

\newcommand{\B}{{\cal B}}

\newcommand{\D}{{\cal D}}

\newcommand{\A}{{\cal A}}
\newcommand{\M}{{\cal M}}

\newcommand{\Competence}{\triangledown}

\definecolor{gray}{RGB}{127,127,127}

\newcommand{\factorU}{\Competence}

\newcommand{\avg}[1]{\oplus_{#1}}

\renewcommand{\phi}{\varphi}

\newcommand{\maxs}[1]{\max\set{#1}}

\newcommand{\Paragraph}[1]{\paragraph*{#1.}}
\newcommand{\ST}{ : \:}

\newcommand{\DLTL}{\ensuremath{\LTL^{\text{disc}}[\D]}}
\newcommand{\DLTLE}{\ensuremath{\LTL^{\text{disc}}[E]}}
\newcommand{\LTLD}{\DLTL}
\newcommand{\DPLTL}{\ensuremath{\PLTL^{\text{disc}}[\D]}}
\newcommand{\AvDLTL}{\ensuremath{\LTL^{\text{disc$\avg{}$}}[\D]}}

\newcommand{\sep}{~ | ~}

\newcommand{\df}{\eta}
\newcommand{\pos}[1]{{#1}^{+}}
\newcommand{\notone}[1]{{#1}^{<1}}

\newcommand{\spb}{\hspace*{-0.42cm}}

\begin{document}

\belowdisplayskip=3pt
\abovedisplayskip=3pt
\title{Discounting in LTL}
\titlerunning{Discounting in LTL}
\author{Shaull Almagor\inst{1} \and Udi Boker\inst{2} \and Orna Kupferman\inst{1}}
\institute{The Hebrew University, Jerusalem, Israel. \and The Interdisciplinary Center, Herzliya, Israel.}
\maketitle

\begin{abstract}
In recent years, there is growing need and interest in formalizing and reasoning about the quality of software and hardware systems. As opposed to traditional verification, where one handles the question of whether a system satisfies, or not, a given specification, reasoning about quality addresses the question of \emph{how well} the system satisfies the specification. One direction in this effort is to refine the ``eventually'' operators of temporal logic to {\em discounting operators}: the satisfaction value of a specification is a value in $[0,1]$, where the longer it takes to fulfill eventuality requirements, the smaller the satisfaction value is.

In this paper we introduce an augmentation by discounting of Linear Temporal Logic (LTL), and study it, as well as its combination with propositional quality operators. We show that one can augment LTL with an arbitrary set of discounting functions, while preserving the decidability of the model-checking problem. Further augmenting the logic with unary propositional quality operators preserves decidability, whereas adding an average-operator makes some problems undecidable.
%the model-checking problem undecidable. 
We also discuss the complexity of the problem, as well as various extensions.
\end{abstract}

\section{Introduction}
\label{sec:introduction}
One of the main obstacles to the development of complex hardware and software systems lies in ensuring their correctness. A successful paradigm addressing this obstacle is {\em temporal-logic model checking\/} -- given a mathematical model of the system and a temporal-logic formula that specifies a desired behavior of it, decide whether the model satisfies the formula \cite{CGP99}.
Correctness is Boolean: a system can
either satisfy its specification or not satisfy it. The richness
of today's systems, however, justifies specification formalisms that
are {\em multi-valued}. The multi-valued setting arises directly in systems
with quantitative aspects (multi-valued / probabilistic / fuzzy) \cite{DR09b,DV12,FLS08,Kwi07,MLL04}, but is applied also with respect to Boolean systems, where it origins from the semantics of the specification formalism itself 
\cite{ABK13,AFHMS05}.

When considering the \emph{quality} of a system,
satisfying a specification should no longer be a yes/no matter. Different ways of satisfying a specification should induce different levels of quality, which should be reflected in the output of the verification procedure. Consider for example the specification
$\Alw({\it  request} \rightarrow \Ev({\it response\_grant} \vee {\it response\_deny}))$ (``every request is eventually responded, with either a grant or a denial''). There should be a difference between a computation that satisfies it with responses generated soon after requests and one that satisfies it with long waits. Moreover, there may be a difference between grant and deny responses, or cases in which no request is issued.
%shaullPOPL2
The issue of generating high-quality hardware and software systems attracts a lot of attention \cite{Kan02,Spi06}. Quality, however, is traditionally viewed as an art, or as an amorphic ideal.  
%orna3
In \cite{ABK13}, we introduced an approach for formalizing quality. Using it, a user can specify quality formally, according to the importance he gives to components such as security, maintainability, runtime, and more, and then can formally reason about the quality of software.

As the example above demonstrates, we can distinguish between two 
%orna3 (we use components for something else)
aspects 
%components
of the quality of satisfaction. The first, to which we refer as ``temporal quality'' concerns the waiting time to satisfaction of eventualities. The second, to which we refer as ``propositional quality'' concerns prioritizing related components of the specification. 
%orna3
Propositional quality was studied in \cite{ABK13}. In this paper we study temporal quality as well as the combinations of both aspects. One may try to reduce temporal quality to propositional quality by a repeated use of the $\Next$ (``next") operator or by a use of  bounded (prompt) eventualities \cite{AHK10,BC06}. Both approaches, however, partitions the future into finitely many zones and are limited: correctness of $\LTL$ is Boolean, and thus has inherent dichotomy between satisfaction and dissatisfaction. On the other hand, the distinction between ``near'' and ``far'' is not dichotomous.
%orna3 (shorten)
\stam{
One may try to reduce ``temporal quality'' to ``propositional quality'', using the fact that an eventuality involves a repeated choice between satisfying it in the present or delaying its satisfaction to the strict future. This attempt, however, requires unboundedly many applications of the propositional choice, and is similar to a repeated use of the $\Next$ (``next") operator rather than a use of eventuality operators.
Repeated use of $\Next$ is a limited solution, 
as it partitions the future into finitely many zones, all of which are in the ``near future'', except for a single, unbounded, ``far future''. A more involved approach to distinguish between the ``near'' and ``far'' future  includes bounded (prompt) eventualities \cite{AHK10,BC06}. There, one distinguishes between eventualities whose waiting time is bounded and ones that have no bound.

The weakness of both approaches is not surprising -- correctness of $\LTL$ is Boolean, and thus has inherent dichotomy between satisfaction and dissatisfaction. The distinction between ``near'' and ``far'', however, is not dichotomous.
} %of stam
This suggests that in order to formalize temporal quality, one must extend $\LTL$ to an unbounded setting. Realizing this, researchers have suggested to augment temporal logics with {\em future discounting} \cite{AHM03}. In the discounted setting, the satisfaction value of specifications is a numerical value, and it depends, according to some discounting function, on the time waited for eventualities to get satisfied.

In this paper we add discounting to Linear Temporal Logic (\LTL), and study it, as well as its combination with propositional quality operators.
We introduce $\DLTL$ -- an augmentation by discounting of \LTL. The logic $\DLTL$ is actually a family of logics, each parameterized by a set $\D$ of discounting functions --  strictly decreasing functions from $\Nat$ to $[0,1]$ that tend to $0$
(e.g., linear decaying, exponential decaying, etc.). $\DLTL$ includes a discounting-``until'' ($\Until_{\eta}$) 
%and discounting-``dual until''  ($\dUntil_{\eta}$) operators,%
operator, parameterized by a function $\eta \in \D$. We solve the model-checking threshold problem for $\DLTL$: given a Kripke structure $\K$, an $\DLTL$ formula $\varphi$ and a threshold $t\in[0,1]$, the algorithm decides whether the satisfaction value of $\varphi$ in $\K$ is at least $t$.

In the Boolean setting, the automata-theoretic approach has proven to be very useful in reasoning about \LTL specifications. The approach is based on translating \LTL formulas to nondeterministic B\"uchi automata on infinite words \cite{VW86b}.
Applying this approach to the discounted setting, which gives rise to infinitely many satisfaction values, poses a big algorithmic challenge: model-checking algorithms, and in particular those that follow the automata-theoretic approach, are based on an exhaustive search, which cannot be simply applied when the domain becomes infinite.
A natural relevant extension to the automata-theoretic approach is to translate formulas to {\em weighted automata} \cite{Moh97}. Unfortunately, these extensively-studied models are complicated and many problems become undecidable for them \cite{Kro94}.
We show that for threshold problems, we can translate $\DLTL$ formulas into (Boolean) nondeterministic B\"uchi automata, with the property that the automaton accepts a lasso computation iff the formula attains a value above the threshold on that computation.
Our algorithm relies on the fact that the language of an automaton is non-empty iff there is a lasso witness for the non-emptiness.
We cope with the infinitely many possible satisfaction values by using the discounting behavior of the eventualities and the given threshold in order to partition the state space into a finite number of classes.
The complexity of our algorithm depends on the discounting functions used in the formula.
We show that for standard discounting functions, such as exponential decaying, the problem is PSPACE-complete -- not more complex than standard \LTL. 
%shaull2
The fact our algorithm uses Boolean automata also enables us to suggest a solution for threshold satisfiability, and to give a partial solution to threshold synthesis. In addition, it allows to adapt the heuristics and tools that exist for Boolean automata.

Before we continue to describe our contribution, let us review existing work on discounting.
The notion of discounting has been studied in several fields, such as economy, game-theory, and Markov decision processes \cite{Sha53}. In the area of formal verification, it was suggested in
\cite{AHM03} to augment the $\mu$-calculus with discounting operators. The discounting suggested there is exponential; that is, with each iteration, the satisfaction value of the formula decreases by a multiplicative factor in $(0,1]$. Algorithmically, \cite{AHM03} shows how to evaluate discounted $\mu$-calculus formulas with arbitrary precision. Formulas of $\LTL$ can be translated to the $\mu$-calculus, thus \cite{AHM03} can be used in order to approximately model-check discounted-$\LTL$ formulas. However, the translation from $\LTL$ to the $\mu$-calculus involves an exponential blowup \cite{Dam94} (and is complicated), making this approach inefficient. Moreover, our approach allows for arbitrary discounting functions, and the algorithm returns an exact solution to the threshold model-checking problem, which is more difficult than the approximation problem.

Closer to our work is \cite{AFHMS05}, where $\CTL$ is augmented with discounting and weighted-average operators.  The motivation in~\cite{AFHMS05} is to introduce a logic whose semantics is not too sensitive to small perturbations in the model. Accordingly, formulas are evaluated on weighted-systems or on Markov-chains. Adding discounting and weighted-average operators to $\CTL$ preserves its appealing complexity, and the model-checking problem for the augmented logic can be solved in polynomial time. As is the case in the traditional, Boolean, semantics, the expressive power of discounted $\CTL$ is limited.
%orna1
The fact the same combination, of discounting and weighted-average operators, leads to undecidability in the context of LTL witnesses the technical challenges of the $\DLTL$ setting.  

Perhaps closest to our approach is~\cite{Man12}, where a version of discounted-\LTL was introduced. Semantically, there are two main differences between the logics. The first is that ~\cite{Man12} uses discounted sum, while we interpret discounting without accumulation, and the second is that the discounting there replaces the standard temporal operators, so all  eventualities are discounted. As discounting functions tend to $0$, this strictly restricts the expressive power of the logic, and one cannot specify traditional eventualities in it. On the positive side, it enables a clean algebraic characterization of the semantics, and indeed the contribution in \cite{Man12} is a comprehensive study of the mathematical properties of the logic. Yet, ~\cite{Man12} does not study algorithmic questions
about to the logic. We, on the other hand, focus on the algorithmic properties of the logic, and specifically on the model-checking problem.

Let us now return to our contribution. After introducing $\DLTL$ and studying its model-checking problem, we augment $\DLTL$ with propositional quality operators. Beyond the operators $\min$, $\max$, and $\neg$, which are already present, two basic propositional quality operators are the multiplication of an $\DLTL$ formula by a constant in $[0,1]$, and the averaging between the satisfaction values of two $\DLTL$ formulas \cite{ABK13}. We show that while the first extension does not increase the expressive power of $\DLTL$ or its complexity, the latter causes 
%the model-checking 
some problems (e.g. validity) to become undecidable. In fact, things become undecidable even if we allow averaging in combination with a single discounting function. Recall that this is in contrast with the extension of discounted $\CTL$ with an average operator, where the complexity of the model-checking problem stays polynomial \cite{AFHMS05}.

We consider additional extensions of $\DLTL$. First, we study a variant of the discounting-eventually operators in which we allow the discounting
%function
to tend to arbitrary values in $[0,1]$ (rather than to $0$). This captures the intuition that we are not always pessimistic about the future, but can be, for example, ambivalent about it, by tending to $\frac{1}{2}$. We show that all our results hold under this extension.
Second, we add to $\DLTL$ {\em past\/} operators and their discounting versions (specifically, we allow a discounting-``since" operator, and its dual). In the traditional semantics, past operators enable clean specifications of many interesting properties, make the logic exponentially more succinct, and can still be handled within the same complexity bounds \cite{LS94,LPZ85}. We show that the same holds for the discounted setting.
Finally, we show how $\DLTL$ and algorithms for it can be used also for reasoning about weighted systems.

Due to lack of space, the full proofs appear in the appendix.
%shaullPOPL2
%To improve readability, some of the full proofs appear in the appendix.

\section{The Logic \DLTL}
The linear temporal logic $\DLTL$ generalizes $\LTL$ by adding discounting temporal operators. The logic is actually a family of logics, each parameterized by a set $\D$ of discounting functions.

Let $\Nat=\set{0,1,...}$. A function $\df: \Nat\to [0,1]$ is a {\em discounting function} if $\lim_{i\to \infty}\df(i)=0$, and $\df$ is strictly monotonic-decreasing. Examples for natural discounting functions are $\df(i)=\lambda^i$, for some $\lambda\in (0,1)$, and $\df(i)=\frac{1}{i+1}$.

Given a set of discounting functions $\D$, we define the logic $\DLTL$ as follows. The syntax of $\DLTL$ adds to $\LTL$ the operator $\phi\Until_\df \psi$ (discounting-Until),
%shaull3
%operators $\phi\Until_\df \psi$ (discounting-Until) and $\phi \dUntil_\df \psi$ (discounting-dual Until),
for every function $\df\in \D$.
Thus, the syntax is given by the following grammar, where $p$ ranges over the set $AP$ of atomic propositions and $\df\in \D$.
$$\phi:= \True \sep p \sep \neg \phi \sep \phi\vee \phi \sep  \Next \phi \sep \phi\Until \phi \sep \phi\Until_\df \phi.$$
%shaull2 \sep \phi\dUntil_\df \phi.$$

The semantics of $\DLTL$ is defined with respect to a {\em computation\/} $\pi=\pi^0,\pi^1,\ldots \in (2^{AP})^\omega$. Given a computation $\pi$ and an $\DLTL$ formula $\varphi$, the truth value of $\varphi$ in $\pi$ is a value in $[0,1]$, denoted $\sem{\pi,\phi}$. The value is defined by induction on the structure of $\varphi$ as follows, where $\pi^i=\pi_i,\pi_{i+1},\ldots$.
%\vspace{-0.5mm}
\begin{itemize}
\item $\sem{\pi,\True}=1$. \hspace{2.98cm} \labelitemi ~$\sem{\pi, \phi\vee \psi}=\maxs{\sem{\pi, \phi},\sem{\pi,\psi}}$.
\item $\sem{\pi,p}=
\begin{cases}
1 & \mbox{ if }p\in \pi_0,\\
0 & \mbox{ if } p\notin \pi_0.
\end{cases}$  \hspace{1.5cm} \labelitemi ~$\sem{\pi,\neg \phi}=1-\sem{\pi, \phi}$.
\item $\sem{\pi,\Next \phi}=\sem{\pi^1,\phi}$.
\item $\sem{\pi, \varphi \Until \psi}  = \sup\limits_{i\ge 0} \{ \min \{\sem{\pi^i,\psi},  \min\limits_{0\leq j < i}\{ \sem{\pi^j,\varphi}\}\}\}$.
\item $\sem{\pi, \varphi \Until_{\df} \psi}  = \sup\limits_{i\ge 0} \{ \min \{\df(i)\sem{\pi^i,\psi},  \min\limits_{0\leq j < i}\{\df(j) \sem{\pi^j,\varphi}\}\}\}$.
%shaull2
%\item $\sem{\pi, \varphi \dUntil_{\df} \psi}  = \inf\limits_{i\ge 0} \{ \max \{\df(i)\sem{\pi^i,\psi},  \max\limits_{0\leq j < i}\{\df(j) \sem{\pi^j,\varphi}\}\}\}$.

\end{itemize}

%orna2: remove \ut and add instead a footnote at the end of this paragraph saying that one could think about additional operators, where the value of future events is discounted towards any value in [0,1] and not necessarily decreased to 0. 
%shaull2
The intuition is that events that happen in the future have a lower influence,
 and the rate by which this influence decreases depends on the function $\df$.
\footnote{Observe that in our semantics the satisfaction value of future events tends to $0$. One may think of scenarios where future events are discounted towards another value in $[0,1]$ (e.g. discounting towards $\frac12$ as ambivalence regarding the future). We address this in Section~\ref{sec:DiscountingTendency}.} For example, the satisfaction value of a formula $\varphi \Until_{\df} \psi$ in a computation $\pi$ depends on the best (supremum) value that $\psi$ can get along the entire computation, while considering the discounted satisfaction of $\psi$ at a position $i$, as a result of multiplying it by $\df(i)$,
%orna1
and the same for the value of $\varphi$ in the prefix leading to the $i$-th position.

We add the standard abbreviations $\Ev \phi\equiv \True\Until \phi$, and $\Alw \phi=\neg \Ev \neg \phi$, as well as their quantitative counterparts: $\Ev_\df \phi\equiv \True\Until_\df \phi$, and $\Alw_\df \phi=\neg \Ev_\df \neg \phi$.
%shaull2
We denote by $|\phi|$ the number of subformulas of $\phi$. 

A computation of the form $\pi=u\cdot v^\omega$, for $u,v\in (2^{AP})^*$, with $v\neq \epsilon$, is called a {\em lasso computation}. 
%orna2 If we have no room, this sentence can be removed - it will appear later.
We observe that since a specific lasso computation has only finitely many distinct suffixes, the $\inf$ and $\sup$ in the semantics of $\DLTL$ can be replaced with $\min$ and $\max$, respectively, when applied to lasso computations.

The semantics is extended to {\em Kripke structures} by taking the path that admits the lowest satisfaction value. Formally, for a Kripke structure $\K$ and an $\DLTL$ formula $\phi$ we have that $\sem{\K,\phi}=\inf\set{\sem{\pi,\phi}\ST \pi \text{ is a computation of $\K$} }$.

\begin{example}
Consider a lossy-disk: every moment in time there is a chance that some bit would flip its value. Fixing flips is done by a global error-correcting procedure. This procedure manipulates the entire content of the disk, such that initially it causes more errors in the disk, but the longer it runs, the more bits it fixes. 

Let ${\it init}$ and ${\it terminate}$ be atomic propositions indicating when the error-correcting procedure is initiated and terminated, respectively.
The quality of the disk (that is, a measure of the amount of correct bits) can be specified by the formula $\phi=\Alw\Ev_\eta({\it init} \wedge \neg \Ev_\mu {\it terminate})$ for some appropriate discounting functions $\eta$ and $\mu$. 
Intuitively, $\phi$ gets a higher satisfaction value the shorter the waiting time is between initiations of the error-correcting procedure, and the longer the procedure runs (that is, not terminated) in between these initiations. Note that the ``worst case'' nature of $\LTLD$ fits here. For instance, running the procedure for a very short time, even once, will cause many errors.
\end{example}

\section{$\DLTL$ Model Checking}
\label{sec:alg proc}
In the Boolean setting, the model-checking problem asks, given an $\LTL$ formula $\phi$ and a Kripke structure $\K$, whether $\sem{\K,\phi}=\True$. In the quantitative setting, the model-checking problem is to compute $\sem{\K,\phi}$, where $\phi$ is now an $\DLTL$ formula. A simpler version of this problem is the threshold model-checking problem: given $\phi$, $\K$, and a threshold $v\in [0,1]$, decide whether $\sem{\K,\phi}\ge v$. In this section we show how we can solve the latter.
%, as well as other decision problems for $\DLTL$.

Our solution uses the automata-theoretic approach, and consists of the following steps. We start by translating $\phi$ and $v$ to an alternating weak automaton $\A_{\phi,v}$ such that $L(\A_{\phi,v})\neq \emptyset$ iff there exists a computation $\pi$ such that $\sem{\pi,\phi}>v$.
%such that for every computation $\pi$, it holds that $\sem{\pi,\phi}>v$ iff $\A_{\phi,v}$ accepts $\pi$. 
The challenge here is that $\phi$ has infinitely many satisfaction values, naively implying an infinite-state automaton. We show that using the threshold and the discounting behavior of the eventualities, we can restrict attention to a finite resolution of satisfaction values, enabling the construction of a finite automaton.  Complexity-wise, the size of $\A_{\phi,v}$ depends on the functions in $\D$. In Section~\ref{subsec:ExpDiscounting}, we analyze the complexity for the case of exponential-discounting functions.

The second step is to construct a nondeterministic B\"uchi automaton $\B$ that is equivalent to $\A_{\phi,v}$. In general, alternation removal might involve an exponential blowup in the state space \cite{MH84}.
We show, by a careful analysis of $\A_{\phi,v}$, that we can remove its alternation while only having a polynomial state blowup.

We complete the model-checking procedure by composing the nondeterministic B\"uchi automaton $\B$ with the Kripke structure $\K$, as done in the traditional, automata-based, model-checking procedure.

The complexity of model-checking an $\LTLD$ formula depends on the discounting functions in $\D$. Intuitively, the faster the discounting tends to 0, the less states there will be. For exponential-discounting, we show that the complexity is NLOGSPACE in the system (the Kripke structure) and PSPACE in the specification (the $\LTLD$ formula and the threshold), staying in the same complexity classes of standard \LTL model-checking.

We conclude the section by showing how to use the generated nondeterministic B\"uchi automaton for 
%shaull2
addressing threshold satisfiability and synthesis.

\subsection{Alternating Weak Automata}

%short
%Given an alphabet $\Sigma$, an {\em infinite word over $\Sigma$} is an
%infinite sequence $w=\sigma_0 \cdot \sigma_1 \cdots$
%of letters in $\Sigma$.
For a given set $X$, let ${\cal B}^{+}(X)$ be the set of positive Boolean formulas over $X$ (i.e., Boolean formulas built from elements in $X$
using $\wedge$ and $\vee$), where we also allow the formulas $\True$ and $\False$. For $Y \subseteq X$, we say that $Y$ {\em satisfies\/} a
formula $\theta \in {\cal B}^+(X)$ iff the truth assignment that assigns {\em true} to the members of $Y$ and assigns {\em false} to the
members of $X \setminus Y$ satisfies $\theta$.
An {\em alternating B\"uchi automaton on infinite words\/} is a tuple
$\A=\tuple{\Sigma, Q, q_{in}, \delta, \alpha}$, where $\Sigma$
is the input alphabet, $Q$ is a finite set of states, $q_{in} \in Q$ is an initial state, $\delta : Q \times \Sigma \rightarrow {\cal B}^+(Q) $
is a transition function, and $\alpha \subseteq Q$ is a set of accepting states. We define runs of $\A$ by means of (possibly) infinite {\sc dag}s
(directed acyclic graphs).  A run of $\A$ on a word
$w=\sigma_0 \cdot \sigma_1 \cdots \in \Sigma^\omega$ is a (possibly) infinite {\sc dag}  ${\cal G}=\pair{V,E}$ satisfying the following (note that there may be several runs
of $\A$ on $w$).
\begin{itemize}
\item
$V \subseteq Q \times \Nat$ is as follows. Let $Q_l \subseteq Q$ denote
 all states
 in level $l$. Thus, $Q_l=\{q: \pair{q,l} \in V\}$. Then,
$Q_0=\{q_{in}\}$, and $Q_{l+1}$ satisfies $\bigwedge_{q \in Q_l}
 \delta(q,\sigma_l)$.
\item For every $l\in \Nat$, $Q_l$ is minimal with respect to containment.
\item
%shaull3
%$E \subseteq \bigcup_{l \geq 0} (Q_l \times \{l\}) \times (Q_{l+1} \times \{l+1\})$ is such that $E(\pair{q,l},\pair{q',l+1})$ iff $Q_{l+1} \setminus \{q'\}$ does not satisfy $\delta(q,\sigma_l)$.
$E \subseteq \bigcup_{l \geq 0} (Q_l \times \{l\}) \times (Q_{l+1} \times \{l+1\})$ is such that for every state $q\in Q_l$, the set $\set{q'\in Q_{l+1}: E(<q,l>,<q',l+1>)}$ satisfies $\delta(q,\sigma_l)$.
\end{itemize}
Thus, the root of the {\sc dag} contains the initial state of the automaton, and the states associated with nodes in level $l+1$ satisfy the
transitions from states corresponding to nodes in level $l$.
The run ${\cal G}$ accepts the word $w$ if all its infinite paths satisfy the acceptance condition $\alpha$. Thus, in the case of B\"uchi automata,
all the infinite paths have infinitely many nodes $\zug{q,l}$ such that $q\in \alpha$
(it is not hard to prove that every infinite path in ${\cal G}$ is part of an infinite path starting in level $0$).
A word $w$ is accepted by $\A$ if there is a run that accepts it. The language of ${\cal A}$,
denoted $L({\cal A})$, is the set of infinite words that ${\cal A}$ accepts.

When the formulas in the transition function of $\A$ contain only disjunctions, then $\A$ is  nondeterministic, and its runs are {\sc dag}s of width 1, where at each level there is a single node. %Accordingly, we sometimes refer to the transition function of a nondeterministic automaton as $\delta:Q \times \Sigma \rightarrow 2^Q$, and refer to its runs as sequences $r=q_0,q_1, \ldots$ of states.

The alternating automaton $\A$ is {\em weak\/}, denoted $\AWW$, if its state space $Q$ can be partitioned into sets $Q_1,\ldots,Q_k$, such that the following hold: First, for every $1\le i\le k$ either $Q_i\subseteq \alpha$, in which case we say that $Q_i$ is an accepting set,  or $Q_i\cap \alpha=\emptyset$, in which case we say that $Q_i$ is rejecting. Second, there is a partial-order $\leq$ over the sets, and for every $1\le i, j\le k$, if $q\in Q_i$, $s\in Q_j$, and $s \in \delta(q,\sigma)$ for some $\sigma\in \Sigma$, then $Q_j \leq Q_i$. Thus, transitions can lead only to states that are smaller in the partial order. Consequently, each run of an $\AWW$ eventually gets trapped in a set $Q_i$ and is accepting iff this set is accepting. 
%A useful property of $\AWW$ is that it is very easy to complement the language of a given automaton: we only have to dualize its transitions (that is, swap disjunctions with conjunctions, and trues with falses) and dualize its acceptance condition (that is, swap accepting and rejecting sets).

\subsection{From $\DLTL$ to $\AWW$}\label{sec:DltlToAww}
Our model-checking algorithm is based on translating an $\DLTL$ formula $\phi$ to an $\AWW$.
Intuitively, the states of the $\AWW$ correspond to assertions of the form $\psi>t$ or $\psi<t$ for every subformula $\psi$ of $\phi$, and for certain thresholds $t\in [0,1]$. 
A lasso computation is then accepted from state $\psi>t$ iff $\sem{\pi,\psi}>t$. The assumption about the computation being a lasso is needed only for the ``only if'' direction, and it does not influence the proof's generality since the language of an automaton is non-empty iff there is a lasso witness for its non-emptiness. By setting the initial state to $\phi>v$, we are done. 

Defining the appropriate transition function for the $\AWW$ follows the semantics of $\LTLD$ in the expected manner. A naive construction, however, yields an infinite-state automaton (even if we only expand the state space on-the-fly, as discounting formulas can take infinitely many satisfaction values).
As can be seen in the proof of Theorem~\ref{thm:DLTL to AWW}, the ``problematic'' transitions are those that involve the discounting operators. The key observation is that, given a threshold $v$ and a computation $\pi$, when evaluating a discounted operator on $\pi$, one can restrict attention to two cases: either the satisfaction value of the formula goes below $v$, in which case this happens after a bounded prefix, or the satisfaction value always remains above $v$, in which case we can replace the discounted operator with a Boolean one. This observation allows us to expand only a finite number of states on-the-fly.

Before describing the construction of the $\AWW$, we need the following lemma, which reduces an extreme satisfaction of an $\DLTL$ formula, meaning satisfaction with a value of either $0$ or $1$, to a Boolean satisfaction of an LTL formula. The proof proceeds by induction on the structure of the formulas.
%orna2: add that |\pos{\phi}| and |\notone{\phi}| = O(|\phi|)
%shaull2
\begin{lemma}
\label{lem:LTL for positive}
Given an $\DLTL$ formula $\phi$, there exist $\LTL$ formulas $\pos{\phi}$ and $\notone{\phi}$ such that $|\pos{\phi}|$ and $|\notone{\phi}|$ are both $O(|\phi|)$ and the following hold for every computation $\pi$.
\begin{enumerate}
\item If $\sem{\pi,\phi}>0$ then $\pi\models \pos{\phi}$, and if $\sem{\pi,\phi}<1$ then $\pi\models \notone{\phi}$.
\item If $\pi$ is a lasso, then if $\pi\models \pos{\phi}$ then $\sem{\pi,\phi}> 0$ and if $\pi\models \notone{\phi}$ then $\sem{\pi,\phi}<1$.
\end{enumerate} 
\end{lemma}

Henceforth, given an $\LTLD$ formula $\phi$, we refer to $\pos{\phi}$ as in Lemma~\ref{lem:LTL for positive}.

Consider an $\DLTL$ formula $\phi$. By Lemma~\ref{lem:LTL for positive}, if there exists a computation $\pi$ such that $\sem{\pi,\phi}>0$, then $\pos{\phi}$ is satisfiable. Conversely, since $\pos{\phi}$ is a Boolean $\LTL$ formula, then by \cite{Var96} we know that $\pos{\phi}$ is satisfiable iff there exists a lasso computation $\pi$ that satisfies it, in which case $\sem{\pi,\phi}>0$. We conclude with the following.
\begin{corollary}
\label{cor:satisfaction}
Consider an $\LTLD$ formula $\phi$. There exists a computation $\pi$ such that $\sem{\pi,\phi}>0$ iff there exists a lasso computation $\pi'$ such that $\sem{\pi',\phi}>0$, in which case $\pi'\models \pos{\phi}$ as well.
\end{corollary}

\begin{remark}
The curious reader may wonder why we do not prove that $\sem{\pi,\phi}>0$ iff $\pi\models\pos{\phi}$ for every computation $\pi$. As it turns out, a translation that is valid also for computations with no period is not always possible. For example, as is the case with the prompt-eventuality operator of \cite{KPV08}, the formula $\phi=\Alw(\Ev_\eta p)$ is such that the set of computations $\pi$ with $\sem{\pi,\phi}>0$ is not $\omega$-regular, thus one cannot hope to define an \LTL formula $\pos{\phi}$.
\end{remark}

%\Paragraph{The Construction}

We start with some definitions. For a function $f:\Nat\to [0,1]$ and for $k\in \Nat$, we define $f^{+k}:\Nat\to [0,1]$ as follows. For every $i\in \Nat$ we have that $f^{+k}(i)=f(i+k)$. %That is, we delay $f$ by an offset of $k$ steps.

Let $\phi$ be an $\DLTL$ formula over $AP$.
We define the {\em extended closure} of $\phi$, denoted $xcl(\phi)$, to be the set of all the formulas $\psi$ of the following {\em classes}:
\begin{enumerate}
\item $\psi$ is a subformula of $\phi$.
\item $\psi$ is a subformula of $\pos{\theta}$ or $\pos{\neg \theta}$, where $\theta$ is a subformula of $\phi$. 
%orna2: I removed the footnote since we say it after the lemma.
%shaull2: this is not entirely justified, since after Lemma 2 we say that phi^+ is what we mean in the Lemma, but in the footnote we say that among all the formulas with the property of \phi^+, we are taking the one from the *proof* of Lemma 2.
%\footnote{
%The reader may be concerned that $\pos{\theta}$ is not uniquely defined. While this does not affect the proof, one may assume that the construction is as in the proof of Lemma~\ref{lem:LTL for positive}.}
\item $\psi$ is of the form $\theta_1\Until_{\df^{+k}} \theta_2$ for $k\in \Nat$, where $\theta_1\Until_{\df}\theta_2$ is a subformula of $\phi$.
\end{enumerate}
Observe that $xcl(\phi)$ may be infinite, and that it has both $\DLTL$ formulas (from Classes $1$ and $3$) and $\LTL$ formulas (from Class $2$).

\begin{theorem}
\label{thm:DLTL to AWW}
Given an $\DLTL$ formula $\phi$ and a threshold $v\in [0,1]$, there exists an $\AWW$ $\A_{\phi,v}$ such that for every computation $\pi$ the following hold.
\begin{enumerate}
\item If $\sem{\pi,\phi}>v$, then $\A_{\phi,v}$ accepts $\pi$.
\item If $\A_{\phi,v}$ accepts $\pi$ and $\pi$ is a lasso computation, then $\sem{\pi,\phi}> v$.
\end{enumerate} 
\end{theorem}
\vspace*{-3mm}
\begin{proof}
We construct $\A_{\phi,v}=\zug{Q,2^{AP},Q_0,\delta,\alpha}$ as follows.
%such that for every computation $\pi$ it holds that $\A_{\phi,v}$ accepts $\pi$ iff $\sem{\pi,\phi}> v$.

The state space $Q$ consists of two types of states. Type-1 states are assertions of the form $(\psi>t)$ or $(\psi<t)$, where $\psi\in xcl(\phi)$ is of Class 1 or 3 and $t\in [0,1]$. Type-2 states correspond to $\LTL$ formulas of Class 2. Let $S$ be the set of Type-1 and Type-2 states for all $\psi\in xcl(\phi)$ and thresholds $t\in[0,1]$. Then, $Q$ is the subset of $S$ constructed on-the-fly according to the transition function defined below. We later show that $Q$ is indeed finite.

The transition function $\delta:Q\times 2^{AP}\to \B^+(Q)$
is defined as follows. For Type-2 states, the transitions are as in the standard translation from $\LTL$ to $\AWW$ \cite{Var96} (see Appendix~\ref{apx: standard LTL to AWW} for details). For the other states, we define the transitions as follows. Let $\sigma\in 2^{AP}$.
\begin{itemize}
\item $\delta((\True>t), \sigma)=\left[\begin{array}{lc}
\True & \text{ if }t<1,\\
\False & \text{ if }t=1.
\end{array}\right.$

\item $\delta((\False>t), \sigma)=\False$.

\item
$\delta((\True<t), \sigma)=\False.$

\item $\delta((\False<t), \sigma)=\left[\begin{array}{ll}
\True & \text{ if }t>0,\\
\False &\text{ if } t=0.
\end{array}\right.$

\item $\delta((p>t),\sigma)=\left[\begin{array}{ll}
\True & \text{ if }p\in \sigma \text{ and } t<1,\\
\False & \text{ otherwise.}
\end{array}\right.$
%\item

\item
$\delta((p<t),\sigma)=\left[\begin{array}{ll}
\False & \text{ if }p\in \sigma \text { or } t=0,\\
\True & \text{ otherwise.}
\end{array}\right.$
\item
$\delta((\psi_1\vee \psi_2 >t),\sigma)=\delta((\psi_1>t),\sigma)\vee
\delta( (\psi_2>t),\sigma)$.
\item
$\delta((\psi_1\vee \psi_2 <t),\sigma)=\delta((\psi_1<t),\sigma)\wedge
\delta( (\psi_2<t),\sigma)$.

\item $\delta((\neg \psi_1 >t),\sigma)=\delta((\psi_1<1-t),\sigma)$
\item $\delta((\neg \psi_1 <t),\sigma)=\delta((\psi_1>1-t),\sigma)$.

\item $\delta((\Next \psi_1>t),\sigma)=(\psi_1>t)$.
 \item $\delta((\Next \psi_1<t),\sigma)=(\psi_1<t)$.

%%%%%%%Bool Until
\item
$\delta((\psi_1\Until \psi_2>t),\sigma)=\left[ \begin{array}{ll}
\delta((\psi_2> t),\sigma)\vee [\delta((\psi_1>t),\sigma)\wedge (\psi_1\Until \psi_2> t)] & \mbox{ if $0<t< 1$},\\
\False &  \mbox{ if $t \geq 1$},\\
\delta((\pos{(\psi_1\Until \psi_2)}),\sigma) & \mbox{ if $t=0$}.
\end{array}
\right.
$

\item
$\delta((\psi_1\Until \psi_2<t),\sigma)=\left[ \begin{array}{ll}
\delta((\psi_2< t),\sigma)\wedge [\delta((\psi_1<t),\sigma)\vee (\psi_1\Until \psi_2< t)] & \mbox{ if $0<t\le 1$},\\
\True &  \mbox{ if $t > 1$},\\
\False & \mbox{ if $t=0$}.
\end{array}
\right.
$

%%%%%%Until
\item
$\delta((\psi_1\Until_\df \psi_2>t),\sigma)=$
 %\spb
$\left[ \begin{array}{ll}
\delta((\psi_2> \frac{t}{\df(0)}),\sigma)\vee [\delta((\psi_1>\frac{t}{\df(0)}),\sigma)\wedge (\psi_1\Until_{\df^{+1}}\psi_2> t)] & \mbox{ if $0<\frac{t}{\df(0)}< 1$},\\
\False &  \mbox{ if $\frac{t}{\df(0)} \geq 1$},\\
\delta((\pos{(\psi_1\Until_\df \psi_2)}),\sigma) & \mbox{ if $\frac{t}{\df(0)}=0$ (i.e.,  $t=0$)}.
\end{array}
\right.
$

\item
$\delta((\psi_1\Until_\df \psi_2<t),\sigma)=$
 %\spb
$\left[ \begin{array}{ll}
\delta((\psi_2<\frac{t}{\df(0)}),\sigma)\wedge [\delta((\psi_1<\frac{t}{\df(0)}),\sigma)\vee (\psi_1\Until_{\df^{+1}}\psi_2<t)] & \mbox{ if $0<\frac{t}{\df(0)}\le 1$},\\
\True &  \mbox{ if $\frac{t}{\df(0)}> 1$},\\
\False & \mbox{ if $\frac{t}{\df(0)}=0$ (i.e., $t=0$)}.
\end{array}
\right.
$
\end{itemize}
We provide some intuition for the more complex parts of the transition function: consider, for example, the transition $\delta((\psi_1\Until_\df \psi_2 > t),\sigma)$. Since $\df$ is decreasing, the highest possible satisfaction value for $\psi_1\Until_\df \psi_2$ is $\df(0)$. Thus, if $\df(0)\le t$ (equivalently, $\frac{t}{\df(0)}\ge 1$), then it cannot hold that $\psi_1\Until_\df \psi_2 >t$, so the transition is to $\False$. If $t=0$, then we only need to ensure that the satisfaction value of $\psi_1\Until_\df \psi_2$ is not $0$. 
To do so, we require that $\pos{(\psi_1\Until_\df \psi_2)}$ is satisfied. By Corollary~\ref{cor:satisfaction}, this is equivalent to the satisfiability of the former. 
So the transition is identical to that of the state $\pos{(\psi_1\Until_\df \psi_2)}$. Finally, if $0<t<\df(0)$, then (slightly abusing notation) the assertion $\psi_1\Until_\df \psi_2 >t$ is true if either $\df(0) \psi_2> t$ is true, or both $\df(0)\psi_1 >t$ and $\psi_1\Until_{\df+1} \psi_2 >t$ are true.

The initial state of $\A_{\phi,v}$ is $(\phi>v)$.
%shaull2 (removed newline)
The accepting states are these of the form $(\psi_1\Until\psi_2<t)$, as well as accepting states that arise in the standard translation of Boolean \LTL to \AWW (in Type-2 states).
Note that each path in the run of $\A_{\phi,v}$ eventually gets trapped in a single state. Thus, $\A_{\phi,v}$ is indeed an $\AWW$. The intuition behind the acceptance condition is as follows. Getting trapped in state of the form $(\psi_1\Until\psi_2<t)$ is allowed, as the eventuality is satisfied with value $0$. On the other hand, getting stuck in other states (or Type-1) is not allowed, as they involve eventualities that are not satisfied in the threshold promised for them.

This concludes the definition of $\A_{\phi,v}$.  Finally, observe that while the construction as described above is infinite (indeed, uncountable), only finitely many states are reachable from the initial state $(\phi>v)$, and we can compute these states in advance.
Intuitively, it follows from the fact that once the proportion between $t$ and $\eta(i)$ goes above $1$, for Type-1 states associated with threshold $t$ and sub formulas with a discounting function $\eta$, we do not have to generate new states.

A detailed proof of $\A$'s finiteness and correctness is given in the appendix.\vspace{-2mm}
\end{proof}

Since $\A_{\phi,v}$ is 
a Boolean automaton, 
then its language is not empty iff it accepts a lasso computation. Combining this observation with Theorem~\ref{thm:DLTL to AWW}, we conclude with the following.
\begin{corollary}
For an $\LTLD$ formula $\phi$ and a threshold $v\in [0,1]$, it holds that $L(\A_{\phi,v})\neq \emptyset$ iff there exists a computation $\pi$ such that $\sem{\pi,\phi}>v$.
\end{corollary}

\subsection{Exponential Discounting}
\newcommand{\dfl}[1]{\exp_{#1}} %{\factorU_{#1}}
\newcommand{\fac}{F}
\label{subsec:ExpDiscounting}
The size of the $\AWW$ generated as per Theorem~\ref{thm:DLTL to AWW} depends on the discounting functions.
In this section, we analyze its size for the class of {\em exponential discounting} functions, showing that it is singly exponential in the specification formula and in the threshold. This class is perhaps the most common class of discounting functions, as it describes what happens in many natural processes (e.g., temperature change, capacitor charge, effective interest rate, etc.) \cite{AHM03,Sha53}.

For $\lambda\in (0,1)$ we define the {\em exponential-discounting} function $\dfl{\lambda}:\Nat\to [0,1]$ by $\dfl{\lambda}(i)=\lambda^i$.
For the purpose of this section, we restrict to $\lambda\in (0,1)\cap \Rat$.
Let $E=\set{\dfl{\lambda}: \lambda\in (0,1)\cap \Rat}$, and consider the logic $\DLTLE$.

For an $\DLTLE$ formula $\phi$ we define the set $\fac(\phi)$ to be $\{\lambda_1,...,\lambda_k \ST$  the operator $\Until_{\dfl{\lambda}} \mbox{ appears in } \phi \}$. Let $|\zug{\phi}|$ be the length of the description of $\phi$. That is, in addition to $|\phi|$, we include in $|\zug{\phi}|$ the length, in bits, of describing $\fac(\phi)$.
\begin{theorem}
\label{thm: exp disc to AWW}
Given an $\DLTLE$ formula $\phi$ and a threshold $v\in [0,1]\cap \Rat$, there exists an $\AWW$ $\A_{\phi,v}$ such that for every computation $\pi$ the following hold.
\begin{enumerate}
\item If $\sem{\pi,\phi}>v$, then $\A_{\phi,v}$ accepts $\pi$.
\item If $\A_{\phi,v}$ accepts $\pi$ and $\pi$ is a lasso computation, then $\sem{\pi,\phi}> v$.
\end{enumerate} 
Furthermore, the number of states of $\A_{\phi,v}$ is singly exponential in $|\zug{\phi}|$ and in the description of $v$.
\end{theorem}
The proof follows from the following observation. Let $\lambda\in (0,1)$ and $v\in (0,1)$. When discounting by $\dfl{\lambda}$, the number of states in the $\AWW$ constructed as per Theorem~\ref{thm:DLTL to AWW} is proportional to the maximal number $i$ such that $\lambda^i>v$, which is at most $\log_\lambda v=\frac{\log v}{\log \lambda}$, which is polynomial in the description length of $v$ and $\lambda$. A similar (yet more complicated) consideration is applied for the setting of multiple discounting functions and negations.

\subsection{From $\A_{\phi,v}$ to an $\NBW$}
Every $\AWW$ can be translated to an equivalent nondeterministic B\"uchi automaton ($\NBW$, for short), yet the state blowup might be exponential 
{BKR10,MH84}. 
By carefully analyzing the $\AWW$ $\A_{\phi,v}$ generated in Theorem~\ref{thm:DLTL to AWW}, we show that it can be translated to an \NBW with only a polynomial blowup.

The idea behind our complexity analysis is as follows.
Translating an $\AWW$ to an $\NBW$ involves alternation removal, which proceeds by keeping track of 
entire 
levels in a run-{\sc dag}. Thus, a run of the $\NBW$ corresponds to a sequence of subsets of $Q$. The key to the reduced state space is that the number of such subsets is only $|Q|^{O(|\phi|)}$ and not $2^{|Q|}$. To see why, consider a subset $S$ of the states of $\A$. We say that $S$ is {\em minimal} if it does not include two states of the form $\phi<t_1$ and $\phi<t_2$, for $t_1<t_2$, nor two states of the form $\phi\Until_{\df^{+i}} \psi <t$ and $\phi\Until_{\df^{+j}}\psi <t$, for $i<j$, and similarly for ``$>$''. Intuitively, sets that are not minimal hold redundant assertions, and can be ignored. Accordingly, we restrict the state space of the $\NBW$ to have only minimal sets.

\vspace*{-3pt}
\begin{lemma}
\label{lem:DLTL to NBW}
For an $\LTLD$ ~formula $\phi$ and $v \in [0,1]$, the $\AWW$ $\A_{\phi,v}$ constructed in Theorem~\ref{thm:DLTL to AWW} with state space $Q$ can be translated to an $\NBW$ with $|Q|^{O(|\phi|)}$ states.
\end{lemma}

%shaullCR

\subsection{Decision Procedures for $\DLTL$}
\Paragraph{Model checking and satisfiability} Consider a Kripke structure $\K$, an $\LTLD$ formula $\phi$, and a threshold $v$. By checking the emptiness of the intersection of $\K$ with
$\A_{\neg\phi,1-v}$, we can solve the threshold model-checking problem. Indeed, $L(\A_{\neg \phi,1-v})\cap L(\K)\neq \emptyset$ iff there exists a lasso computation $\pi$ that is induced by $\K$ such that $\sem{\pi,\phi}<v$, which happens iff it is not true that $\sem{\K,\phi}\ge v$.

%Using Theorem~\ref{thm:DLTL to AWW} and this observation on the formula $\neg\phi$ and the threshold $1-v$, we can also decide whether there exist a path $\pi$ such that $\sem{\pi,\phi}\ge v$, and whether there exists a path $\pi$ such that $\sem{\pi,\phi}<v$.
%Finally, by adjusting the initial states of $\A_{\phi,v}$ and its complement automaton, we can decide whether there exists a path $\pi$ such that $\sem{\pi,\phi}\in I$ for every (continuous) interval $I\subseteq [0,1]$, either closed, open, or half-open.

The complexity of the model-checking procedure depends on the discounting functions in $\D$. For the set of exponential-discounting functions $E$, we provide the following concrete complexities, showing that it stays in the same complexity classes of standard \LTL model-checking.

\begin{theorem}\label{thm:ExpModelCheck}
For a Kripke structure $\K$, an $\DLTLE$ formula $\phi$, and a threshold $v\in [0,1]\cap \Rat$, the problem of deciding whether $\sem{\K,\phi}>v$ is in NLOGSPACE in the number of states of $\K$, and in PSPACE in $|\zug{\phi}|$ and in the description of $v$.
\end{theorem}
\begin{proof}
By Theorem~\ref{thm: exp disc to AWW} and Lemma~\ref{lem:DLTL to NBW},  the size of an $\NBW$ $\B$ corresponding to $\phi$ and $v$ is singly exponential in $|\zug{\phi}|$ and in the description of $v$. Hence, we can check the emptiness of the intersection of $\K$ and $\B$ via standard ``on the fly'' procedures, getting the stated complexities.
\end{proof}

Note that the complexity in Theorem~\ref{thm:ExpModelCheck} is 
only
NLOGSPACE in the system, since our solution does not analyze the Kripke structure, but only takes its product with the specification's automaton. This is in contrast to the approach of model checking temporal logic with (non-discounting) accumulative values, where, when decidable, involves a doubly-exponential dependency on the size of the system~\cite{BCHK11}.

%shaullTACAS14
Finally, observe that the $\NBW$ obtained in Lemma~\ref{lem:DLTL to NBW} can be used to solve the threshold-satisfiability problem: given an $\DLTL$ formula $\phi$ and a threshold $v\in [0,1]$, we can decide whether there is a computation $\pi$ such that $\sem{\pi,\phi}\sim v$, for $\sim \in \{ <, >\}$, and return such a computation when the answer is positive. This is done by simply deciding whether there exists a word that is accepted by the $\NBW$.

\paragraph{Threshold synthesis}
%\begin{remark}
\label{rmk:SAT and synthesis}
Consider an $\LTLD$ formula $\phi$, and assume a partition of the atomic propositions in $\phi$ to input and output signals, we can use the $\NBW$ $\A_{\phi,v}$ in order to address the {\em synthesis} problem, as stated in the following theorem (see Appendix~\ref{apx: synthesis} for the proof).
\begin{theorem}
\label{thm:synthesis}
Consider an $\LTLD$ formula $\phi$. If there exists a transducer $\T$ all of whose computations $\pi$ satisfy $\sem{\pi,\phi}> v$, then we can generate a transducer $\T$ all of whose computations $\tau$ satisfy $\sem{\tau,\phi}\ge v$.
\end{theorem}

%as follows. Recall that the synthesis problem is to generate a transducer $T$ all of whose computations $\pi$ satisfy $\sem{\pi,\phi}>v$, or answer that no such transducer exists. Following standard (Boolean) procedures for synthesis (see \cite{PR89a}), we can generate from $\A_{\phi,v}$ a transducer $T'$ all of whose computations $\pi$ satisfy $\sem{\pi,\phi}\ge v$. (Note that $T'$ differs from $T$, satisfying $\ge v$ rather than $> v$.) This follows by observing that, by the construction of $\A_{\phi,v}$, if it accepts $\pi$ then $\sem{\pi,\phi}\ge v$. 
%(If $\pi$ is a lasso computation, we know that $\sem{\pi,\phi}>v$. Here we observe that otherwise, it is still the case that $\sem{\pi,\phi}\ge v$).
%This solution is only partial, as there might be an $I/O$-transducer whose computations $\pi$ all satisfy $\sem{\pi,\phi}\ge v$, but we will not find it. 
%\end{remark}

\section{Adding Propositional Quality Operators}
\label{sec:adding quality}

As model checking is decidable for $\DLTL$, one may wish to push the limit and extend the expressive power of the logic. In particular, of great interest is the combining of discounting with propositional quality operators \cite{ABK13}.
%One of the ``weaknesses'' of $\DLTL$ is the following. If one considers the generating tree of an $\DLTL$ formula $\phi$, then for every computation $\pi$, the satisfaction value $\sem{\pi,\phi}$ is actually the satisfaction value along a branch of the generating tree (after applying multiplications that stem from the discounting operator). For example, in the formula $p\Until_{\df} (q\vee (r\Until_{\df} s))$, the satisfaction value is obtained by multiplying the discounting factors by the truth values of $q,p,r,s$. That is, the satisfaction value is always a very simple function of the satisfaction values of sub-formulas, namely a multiplication by a scalar of them.
%
%Thus, a possible interesting extension of the logic is to introduce a more complex function. A well-motivated and well-studied example is the average operator $\avg{}$, with the semantics $\sem{\pi,\phi\avg{}\psi}=\frac{\sem{\pi,\phi}+\sem{\pi,\psi}}{2}$.
\subsection{Adding the Average Operator}
\label{sec:average}
A well-motivated extension is the introduction of the average operator $\avg{}$, with the semantics $\sem{\pi,\phi\avg{}\psi}=\frac{\sem{\pi,\phi}+\sem{\pi,\psi}}{2}$.
The work in~\cite{ABK13} proves that extending $\LTL$ by this operator, as well as with other propositional quantitative operators, enables clean specification of quality and results in a logic for which the model-checking problem can be solved in PSPACE.

\newcommand{\inc}{\mbox{\sc inc}\xspace}
\newcommand{\dec}{\mbox{\sc dec}\xspace}
\newcommand{\goto}{\mbox{\sc goto }}
\newcommand{\halt}{\mbox{\sc halt}\xspace}
\newcommand{\jz}[3]{\mbox{\sc if $#1$=0 goto $#2$ else goto $#3$}\xspace}
\newcommand{\CMrun}{\rho}
\newcommand{\comcheck}{\mbox{\rm ComCheck}\xspace}

We show that adding the $\avg{}$ operator to $\DLTL$ gives a logic, denoted $\AvDLTL$, for which the validity problem is undecidable.
%and model-checking problems are undecidable.
%
The validity problem asks, given an $\AvDLTL$ formula $\varphi$ over the atomic propositions $AP$ and a threshold $v\in[0,1]$, whether $\sem{\pi,\varphi}>v$ for every $\pi\in ({2^{AP}})^\omega$.

In the undecidability proof, we show a reduction from the 0-halting problem for two-counter machines.
A {\em two-counter machine} $\M$ is a sequence $(l_1,\ldots,l_n)$ of commands involving two counters $x$ and $y$. We refer to
$\set{1,\ldots,n}$ as the {\em locations} of the machine. There are five possible forms of commands:
$$\inc(c),\ \dec(c),\ \goto l_i,\  \jz{c}{l_i}{l_j},\  \halt,$$
where $c\in \set{x,y}$ is a counter and $1\le i,j\le n$ are locations. Since we can always check whether $c=0$ before a $\dec(c)$ command, we
assume that the machine never reaches $\dec(c)$ with $c=0$. That is, the counters never have negative values. Given a counter machine $\M$,
deciding whether $\M$ halts is known to be undecidable \cite{Min67}. Given $\M$, deciding whether $\M$ halts with both counters having value
$0$, termed the {\em $0$-halting problem}, is also undecidable: given a counter machine $\M$, we can replace every \halt command with a code that clears the counters before
halting. 
%shaullUndNew
In fact, from this we see that the promise problem of deciding whether $\M$ $0$-halts given the promise that either it $0$-halts, or it does not halt at all, is also undecidable.

\begin{theorem}\label{thm:UndecidableAverage}
The validity problem for $\AvDLTL$ is undecidable (for every nonempty set of discounting functions $\D$).
\end{theorem}
\begin{proof}
We start by showing a reduction from the $0$-halting problem for two-counter machines to the following  problem: given an $\AvDLTL$ formula $\varphi$ over the atomic propositions $AP$, whether there exists a path $\pi\in AP^{\omega}$ such that $\sem{\pi,\varphi}\ge\frac12$. We dub this the {\em $\frac12$-co-validity} problem. We will later reduce this problem to the (complement of the) validity problem.

Let $\M$ be a two-counter machine
with commands $(l_1,\ldots,l_n)$. A {\em halting run} of a two-counter machine with commands from the set $L=\{l_1,\ldots,l_n\}$ is a sequence
$\CMrun=\CMrun_1,\ldots,\CMrun_m\in (L\times\Nat \times\Nat)^*$ such that the following hold.
\begin{enumerate}
\item $\CMrun_1=\zug{l_1,0,0}$.
\item For all $1< i\le m$, let $\CMrun_{i-1}=(l_k,\alpha,\beta)$ and $\CMrun_{i}=(l',\alpha',\beta')$. Then, the following hold.
\begin{itemize}
\item If $l_k$ is an $\inc(x)$ command (resp. $\inc(y)$), then $\alpha'=\alpha+1$, $\beta'=\beta$ (resp. $\beta'=\beta+1$, $\alpha'=\alpha$), and $l'=l_{k+1}$.
\item If $l_k$ is a $\dec(x)$ command (resp. $\dec(y)$), then $\alpha'=\alpha-1$, $\beta'=\beta$ (resp. $\beta'=\beta-1$, $\alpha'=\alpha$), and $l'=l_{k+1}$.
\item If $l_k$ is a $\goto l_s$ command, then $\alpha'=\alpha$, $\beta'=\beta$, and $l'=l_s$.
\item If $l_k$ is an $\jz{x}{l_s}{l_t}$ command, then $\alpha'=\alpha$, $\beta'=\beta$, and $l'=l_s$ if $\alpha=0$, and $l'=l_t$ otherwise.
\item If $l_k$ is an $\jz{y}{l_s}{l_t}$ command, then $\alpha'=\alpha,\beta'=\beta$, and $l'=l_s$ if $\beta=0$, and $l'=l_t$ otherwise.

\item If $l'$ is a $\halt$ command, then $i=m$. That is, a run does not continue after $\halt$.
\end{itemize}
\item $\rho_m=\zug{l_k,\alpha,\beta}$ such that $l_k$ is a $\halt$ command.
\end{enumerate}

Observe that the machine $\M$ is deterministic. We say that $\M$ 0-halts if there exists $l\in L$ that is a $\halt$ command, such that the run of $\M$ ends in $\zug{l,0,0}$.

We say that a sequence of commands $\tau\in L^*$ {\em fits} a run $\CMrun$, if $\tau$ is the projection of $\CMrun$ on its first component.

We construct from $\M$ an $\AvDLTL$ formula $\varphi$ such that $\M$ 0-halts iff there exists a computation $\pi$ such that $\sem{\pi,\varphi}\ge \frac12$. The idea behind the construction is as follows. $\phi$ reads computations over the atomic propositions $AP=\set{1,...,n,\#,x,y}$, where $1,...,n$ are the commands of $M$. The computation that $\varphi$ reads corresponds to a description of a run of $\M$, where every triplet $\zug{l_i,\alpha,\beta}$ is encoded as the string $i x^\alpha y^\beta \#$. We ensure that computations that satisfy $\phi$ with value greater than $0$ are such that in every position only a single atomic proposition is true.

%shaull POPL1
\begin{example}
Consider the following machine $\M$:
\begin{enumerate}
\item[$l_1$:] $\inc(x)$
\item[$l_2$:] $\jz{x}{l_6}{l_3}$
\item[$l_3$:] $\inc(y)$
\item[$l_4$:] $\dec(x)$
\item[$l_5$:] $\goto(l_2)$
\item[$l_6$:] $\dec(y)$
\item[$l_7$:] $\halt$
\end{enumerate}
The command sequence that represents the run of this machine is
$$\zug{l_1,0,0},\zug{l_2,1,0},\zug{l_3,1,0},\zug{l_4,1,1},\zug{l_5,0,1},\zug{l_2,0,1},\zug{l_6,0,1},\zug{l_7,0,0}$$
and the encoding of it as a computation is
$$1\# 2x\# 3x\# 4xy\# 5y\# 2y\# 6y\# 7\#$$
\end{example}

The formula $\varphi$ ``states'' (recall that the setting is quantitative, not Boolean) the following properties of the computation $\pi$:
\begin{enumerate}
\item The first configuration in $\pi$ is the initial configuration of $\M$ ($\zug{l_1,0,0}$, or $1\#$ in our encoding).
\item The last configuration in $\pi$ is $\zug{l,0,0}$ (or $k$ in our encoding), where $l$ can be any line whose command is $\halt$.
\item $\pi$ represents a legal run of $\M$, up to the consistency of the counters between transitions.
\item The counters are updated correctly between configurations.
\end{enumerate}
Properties 1-3 are can easily be specified by a $\LTL$ formulas, such that computations which satisfy properties 1-3 get satisfaction value $1$. Property 4 utilizes the expressive power of $\AvDLTL$, as we now demonstrate.
The intuition behind property 4 is the following. We need to compare the value of a counter before and after a command, such that the formula takes a low value if a violation is encountered, and a high value otherwise. Specifically, the formula we construct takes value $\frac12$ if no violation occurred, and a lower value if a violation did occur.

We start with a simpler case, to demonstrate the point. Let $\df\in \D$ be a discounting function.
Consider the formula $CountA:=a\Until_{\df}\neg a$ and the computation $a^ib^j\#^\omega$. It holds that $\sem{a^ib^j,CountA}=\df(i)$. Similarly, it holds that $\sem{a^ib^j\#^\omega,a\Until(b\Until_{\df}\neg b)}=\df(j)$. Denote the latter by $CountB$. Let 
$$CompareAB:=(CountA\avg{}\neg CountB)\wedge (\neg CountA\avg{} CountB).$$ 
We now have that
\begin{align*}
&\sem{a^ib^j\#^\omega,CompareAB} = 
\min\set{\frac{\df(i)+1-\df(j)}{2},\frac{\df(j)+1-\df(i)}{2}}=\frac12-\frac{|\df(i)-\df(j)|}{2}
\end{align*}
and observe that the latter is $\frac12$ iff $i=j$, and is less than $\frac12$ otherwise. This is because $\df$ is strictly decreasing, and in particular an injection.

Thus, we can compare counters. To apply this technique to the encoding of a computation, we only need some technical formulas to ``parse'' the input and find consecutive occurrences of a counter.

We now dive into the technical definition of $\varphi$. The atomic propositions are $AP=\set{1,...,n,\#,x,y}$ (where $l_1,...,l_n$ are the commands of $\M$).
We let
$\varphi:=CheckCmds \wedge CheckInit \wedge CheckFinal\wedge CheckCounters\wedge ForceSingletons$
\vspace*{-3mm}
\paragraph*{ForceSingletons:} This formula ensures that for a computation to get a value of more than $0$, every letter in the computation must be a singleton. Formally,
$$ForceSingletons:=\Alw\left(\bigvee_{p\in AP}(p\wedge \bigwedge_{q\in AP\setminus\set{p}}\neg q)\right).$$

\paragraph*{CheckInit and CheckFinal:}
These formulas check that the initial and final configurations are correct, and that after the final configuration, there are only $\#$s.
$$CheckInit:=1\wedge \Next\#.$$

Let $I=\set{i: l_i=\halt}$, we define
$$CheckFinal:=\Alw((\bigvee_{i\in I} i )\to \Next\Alw \#).$$
Note that $CheckFinal$ also ensures that there counters are 0.

\vspace*{-3mm}
\paragraph*{CheckCmds:} This formula verifies that the local transitions follow the instructions in the counter machine, ignoring the consistency of the counter values, but enforcing that a jump behaves according to the counters.
We start by defining, for every $i\in \set{1,...,n}$, the formula:
$${\it waitfor}(i):=(x\vee y)\Until (\#\wedge \Next i)$$
Intuitively, a computation satisfies this formula (i.e., gets value $1$) iff it reads counter descriptions until the next delimiter, and the next command is $l_i$.

Now, for every $i\in \set{1,...,n}$ we define $\psi_i$ as follows.
\begin{itemize}
\item If $l_i=\goto l_j$, then $\psi_i:=\Next waitfor(j)$.
\item If $l_i\in\set{\inc(c),\dec(c): c\in \set{x,y}}$, then $\psi_i:=\Next waitfor(i+1)$. \footnote{if $i=n$ then this line can be omitted from the initial machine, so w.l.o.g this does not happen.}
\item If $l_i=\jz{x}{l_j}{l_k}$ then
$\psi_i:=\Next((x\to waitfor(k)) \wedge ((\neg x)\to waitfor(j)))$.
\item If $l_i=\jz{y}{l_j}{l_k}$ then
$\psi_i:=\Next(((x\Until y)\wedge waitfor(k)) \vee ((x\Until \#)\wedge waitfor(j)))$.
\item If $l_i=\halt$ we do not really need additional constraints, due to $CheckFinal$. Thus we have $\psi_i=\True$.
\end{itemize}
Finally, we define
$CheckCmds:=\Alw\bigwedge_{i\in \set{1,...,n}}(i\to \psi_i).$

\vspace*{-3mm}
\paragraph*{CheckCounters:} This is the heart of the construction. The formula checks whether consecutive occurrences of the counters match the transition between the commands.
We start by defining
$countX:=x\Until_{\df}\neg x$ and $countY:=x\Until(y\Until_{\df}\neg y)$. Similarly, we have $countX^{-1}=x\Until_{\df}\Next \neg x$ and $countY^{-1}=x\Until(y\Until_{\df} \Next \neg y)$.

We need to define a formula to handle some edge cases. 

Let $I_{\halt}=\set{i: l_i=\halt}$, and similarly define $I_{\dec(x)}$ and $I_{\dec(y)}$. We define
\begin{align*}
&Last:=\bigvee_{i\in I_{\dec(x)}}i\wedge \Next (x\wedge \Next (\#\wedge \Next \bigvee_{k\in I_{\halt}}k))\vee\\
& \bigvee_{i\in I_{\dec(y)}}i\wedge \Next (y\wedge \Next (\#\wedge \Next \bigvee_{k\in I_{\halt}}k))\vee\\
&\bigvee_{i\in \set{1,...,n}}i\wedge \Next (\#\wedge \Next \bigvee_{k\in I_{\halt}}k)
\end{align*}

Intuitively, $Last$ holds exactly in the last transition, that is - before the final 0-halting configuration.

Testing the counters involves six types of comparisons: checking equality, increase by 1, and decrease by 1 for each of the two counters. We define the formulas below for these tests. To explain the formulas, consider for example the formula $Comp(x,=)$. This formula compares the number of $x$'s in the current configuration, with the number of $x$'s in the next configuration. The comparison is based on the comparison we explained above, and is augmented by some parsing, as we need to reach the next configuration before comparing.
\begin{itemize}
\item $Comp(x,=):= \left(countX\avg{}(x\Until(y\Until(\#\wedge \Next\Next \neg countX)))\right)$

$\wedge \left((\neg countX)\avg{}(x\Until(y\Until(\#\wedge \Next\Next countX)))\right)$.

\item $Comp(y,=):= \left(countY\avg{}(x\Until(y\Until(\#\wedge \Next\Next \neg countY)))\right)$

$\wedge \left((\neg countY)\avg{}(x\Until(y\Until(\#\wedge \Next\Next countY)))\right)$

\item $Comp(x,+1):=\left( countX\avg{}(x\Until(y\Until(\#\wedge \Next\Next \neg countX^{-1})))\right)$

$\wedge\left( (\neg countX)\avg{}(x\Until(y\Until(\#\wedge \Next\Next countX^{-1})))\right)$

\item $Comp(y,+1):= \left(countY\avg{}(x\Until(y\Until(\#\wedge \Next\Next \neg countY^{-1})))\right)$

$\wedge \left((\neg countY)\avg{}(x\Until(y\Until(\#\wedge \Next\Next countY^{-1})))\right)$

\item $Comp(x,-1):=\left( countX^{-1}\avg{}(x\Until(y\Until(\#\wedge \Next\Next \neg countX)))\right)$

$\wedge\left( (\neg countX^{-1})\avg{}(x\Until(y\Until(\#\wedge \Next\Next countX)))\right)$

\item $Comp(y,-1):= \left(countY^{-1}\avg{}(x\Until(y\Until(\#\wedge \Next\Next \neg countY)))\right)$

$\wedge \left((\neg countY^{-1})\avg{}(x\Until(y\Until(\#\wedge \Next\Next countY)))\right)$
\end{itemize}

Now, for every $i\in \set{1,...,n}$ we define $\xi_i$ as follows.
\begin{itemize}
\item If $l_i\in \set{\goto l_j,\jz{c}{l_j}{l_k}:c\in \set{x,y}}$ , we need to make sure the value of the counters do not change. We define\\
$\xi_i:=(X(comp(x,=) \wedge comp(y,=))\vee Last.$
\item If $l_i=\inc(x)$, we need to make sure that $x$ increases and $y$ does not change. We define\\
$\xi_i:= (X(comp(x,+1) \wedge comp(y,=))\vee Last.$
\item If $l_i=\inc(y)$, we define\\
$\xi_i:= (X(comp(x,=) \wedge comp(y,+1))\vee Last.$
\item If $l_i=\dec(x)$, we define\\
$\xi_i:=(X(comp(x,-1) \wedge comp(y,=))\vee Last.$
\item If $l_i=\dec(y)$, we define\\
$\xi_i:=(X(comp(x,=) \wedge comp(y,-1))\vee Last.$
\item If $l_i=\halt$ we do not need additional constraints, due to $CheckFinal$. Thus we have $\psi_i=\True$.
\end{itemize}
Finally, we define
$CheckCounters:=\Alw\bigwedge_{i\in \set{1,...,n}}(i\to \xi_i)$

The correctness of the construction is obvious, once one verifies that the defined formulas indeed test what they claim to.
Thus, we conclude that $\M$ 0-halts iff there exists a computation $\pi$ such that $\sem{\pi,\varphi}\ge \frac12$. 

Finally, we reduce the $\frac12$-co-validity problem to the complement of the validity problem: given a formula $\phi$, the reduction outputs $\zug{\neg\phi,\frac12}$. Now, there exists a computation $\pi$ such that $\sem{\pi,\phi}\ge\frac12$ iff there exists a computation $\pi$ such that $\sem{\pi,\neg\phi}\le \frac12$, iff 
it is not true that $\sem{\pi,\neg \phi}>\frac12$ for every computation $\pi$. Thus, $\phi$ is $\frac12$-co-valid iff $\neg\phi$ is not valid for threshold $\frac12$. We conclude that the validity problem is undecidable.

\end{proof}
% % % % % % % % % % % % % %

%shaullUndNew
Studying the proof of Theorem~\ref{thm:UndecidableAverage}, we can actually formulate the reduction more carefully as follows.
\begin{lemma}
\label{lem:reductionProperties}
Given a two-counter machine $\M$ that is promised to either $0$-halt, or not to halt at all, there exists an $\AvDLTL$ formula $\phi$ such that for every computation $\pi$ that represents a computation of $\M$, the following hold.
\begin{enumerate}
\item If $\pi$ is a legal halting computation of $\M$, then $\sem{\pi,\phi}=\frac12$.
\item If $\pi$ cheats in a transition between commands, then $\sem{\pi,\phi}=1$.
\item If $\pi$ cheats in the counter values, then $\sem{\pi,\phi}=\frac12+\epsilon$ such that $\epsilon\ge \frac12(\df(i)-\df(i+1))$ for the minimal difference $\df(i)-\df(i+1)$ where $i$ is a counter value in $\pi$.

%\footnote{We assume for simplicity that the sequence $|\df(i)-\df(i+1)|$ is monotonically decreasing. The proof can be easily adapted to case where this is not true.}
\end{enumerate}

\end{lemma}

We now turn to show that the {\em strict} model-checking problem and the {\em strict} model-checking problem are undecidable for $\AvDLTL$ as well. The strict model-checking problem is to decide, given a Kripke structure $\K$, a formula $\phi$, and a threshold $v$, whether $\sem{\K,\phi}>v$.

%Before proceed to the proof, we present another construction.
%\begin{lemma}
%\label{lem:biggestCounterFormula}
%Given a two-counter machine $\M$ that is promised to either $0$-halt, or not to halt at all, there exists an $\AvDLTL$ formula $\psi$ such that for every computation $\pi$ that represents a computation of $\M$, it holds that $\sem{\pi,\psi}=\frac12 +\frac12\epsilon$, where $\epsilon=\inf(\df(i): i\text{ is the value of a counter in }\pi)$.
%\end{lemma}
%\begin{proof}
%We define 
%%$\psi=\neg (\nabla_\half (F G\eta x or F G\eta y))$.
%%$\psi=\neg \Ev \neg(x\Until_\df \neg x)$
%$\psi=\True\avg{}(\Alw (x\Until_\df \neg x)\wedge \Alw(y\Until_\df \neg y))$.
%It is not hard to verify that $\psi$ has the desired properties.
%
%\end{proof}

\begin{theorem}
\label{thm:strict model checking und}
The strict model-checking problem for $\AvDLTL$ is undecidable (for every nonempty set of Discounting functions $\D$).
\end{theorem}
\begin{proof}
Assume by way of contradiction that the strict model-checking problem is decidable. We show how to decide the $0$-halting promise problem for two-counter machines, thus reaching a contradiction.

Given a two-counter machine $\M$ that is promised to either $0$-halt, or not halt at all, construct the formula $\phi$ as per Lemma~\ref{lem:reductionProperties}, and consider the Kripke structure $\K$ that generates every computation. Observe that by Lemma~\ref{lem:reductionProperties} it holds that $\sem{\K,\phi}\ge \frac12$.

Decide whether $\sem{\K,\phi}>\frac12$. If $\sem{\K,\phi}>\frac12$, then for every computation $\pi$ it holds that $\sem{\pi,\phi}>\frac12$, and by Lemma~\ref{lem:reductionProperties} we conclude that $\M$ does not halt.

If $\sem{\K,\phi}\le \frac12$, then $\sem{\K,\phi}=\frac12$. We observe that there are now two possible cases:
\begin{enumerate}
\item $\M$ halts.
\item $\M$ does not halt, and for every $n$, there are computations that reach $\halt$ while cheating in counter values larger than $n$, and not cheating in the commands.
\end{enumerate}
%Note that cases 1 and 2 are mutually exclusive, since if $\M$ halts, then in the description of the legal run, the counter values are bounded. In order to get case 2, there must exist computations that cheat in the counter values for increasingly large counters, which means that there must be legal runs with unbounded counters, excluding case 1.

We show how to distinguish between cases 1 and 2.

Consider the $\AvDLTL$ formula 
$$\psi=\True\avg{}(\Alw (x\Until_\df \neg x)\wedge \Alw(y\Until_\df \neg y)).$$

It is not hard to verify that for every computation $\pi$ that represents a computation of $\M$, it holds that $\sem{\pi,\psi}=\frac12 +\frac12\epsilon$, where $\epsilon=\inf(\df(i): i\text{ is the value of a counter in }\pi)$.
Let $\xi=\phi\vee \psi$. 

If $\M$ halts (case 1), then for every computation $\pi$ we have one of the following.
\begin{itemize}
\item[a.] $\pi$ describes a legal halting run of $\M$, in which case $\sem{\pi,\phi}=\frac12$ and $\sem{\pi,\psi}>\frac12+\epsilon$ for some $\epsilon>0$ (independent of $\pi$), since the counters are bounded. Thus, $\sem{\pi,\xi}>\frac12+\epsilon$.
\item[b.] $\pi$ cheats in the commands, in which case $\sem{\pi,\phi}=1$, so $\sem{\pi,\xi}=1$.
\item[c.] $\pi$ cheats in the counters, in which case, since the counters are bounded, the first cheat must occur with small counters, and thus $\sem{\pi,\phi}>\frac12+\epsilon$ for some $\epsilon>0$ independent of $\pi$. So $\sem{\pi,\xi}>\frac12+\epsilon$.
\end{itemize}
In all three cases, we get that $\sem{\pi,\xi}>\frac12+\epsilon$, so $\sem{\K,\xi}>\frac12$.

If $\M$ does not halt (case 2), then for every $n$ and for every computation $\pi$ that cheats with counters larger than $n$, it holds that $\sem{\pi,\phi}<\frac12+\epsilon$ where $\epsilon=\epsilon(n)\to 0$ as $n\to\infty$, and since the counters in $\pi$ are large, it also holds that $\sem{\pi,\psi}<\frac12+\epsilon'$, where $\epsilon'=\epsilon'(n)\to 0$ as $n\to\infty$. We conclude that there exists a sequence of computations whose satisfaction values in $\xi$ tend to $\frac12$, and thus $\sem{\K,\xi}=\frac12$.

Thus, in order to distinguish between cases 1 and 2, it is enough to decide whether $\sem{\K,\xi}>\frac12$.  

To conclude, the algorithm for deciding whether $\M$ 0-halts is as follows.
Start by constructing $\phi$. If $\sem{\K,\phi}>\frac12$, then $\M$ does not halt. Otherwise, construct $\xi$. If $\sem{\K,\xi}>\frac12$, then $\M$ halts, and otherwise $\M$ does not halt.
\end{proof}

\begin{theorem}
\label{thm:model checking und}
The model-checking problem for $\AvDLTL$ is undecidable (for every nonempty set of Discounting functions $\D$).
\end{theorem}
\begin{proof}
Recall that for every Kripke structure $\K$ and formula $\phi$ it holds that $
\sem{\K,\phi}\ge v$ iff there does not exist a computation $\pi$ of $\K$ such that $\sem{\pi,\neg \phi}>1-v$. 

We show that the latter problem is undecidable, even if we fix $\K$ to be the system that generates every computation. 

We show a reduction from the $0$-halting promise problem to the latter problem.
Given a two-counter machine $\M$ that is promised to either $0$-halt, or not halt at all, construct the formula $\phi$ as per Lemma~\ref{lem:reductionProperties} and the formula $\psi$ such that for every computation $\pi$ we have that $\sem{\pi,\psi}=\frac12 +\frac12\epsilon$, where $\epsilon=\inf(\df(i)-\df(i+1): i\text{ is the value of a counter in }\pi)$. The formula $\psi$ can be defined as
$$\psi=\True\avg{}(\Alw ((x\vee \Next x)\Until_\df (\neg x\wedge \neg \Next x)\wedge (y\vee \Next y)\Until_\df (\neg y\wedge \neg \Next y))$$.

Let $\theta=(\neg\phi) \avg{} \psi$. We claim that $\M$ halts iff there exists a computation $\pi$ such that $\sem{\pi,\theta}>\frac12$.

If $\M$ halts, then for the computation $\pi$ that describes the halting run of $\M$ it holds that $\sem{\pi,\phi}=\frac12$, and thus $\sem{\pi,\neg\phi}=\frac12$. Since the counters in $\pi$ are bounded (as the run is halting), then $\sem{\pi,\psi}>\frac12$, and thus $\sem{\pi,\theta}>\frac12$.

If $\M$ does not halt, consider a computation $\pi$.
\begin{itemize}
\item If $\pi$ cheats in the commands, then $\sem{\pi,\neg\phi}=0$, so $\sem{\pi,\theta}=0+\frac12\sem{\pi,\psi}\le \frac12$.
\item If $\pi$ cheats in the counters, then $\sem{\pi,\neg\phi}=\frac12-\frac12\epsilon$ and $\sem{\pi,\psi}=\frac12+\frac12\epsilon'$, where $\epsilon\ge \frac12 (\df(i)-\df(i+1))=\epsilon'$ for the smallest difference $\df(i)-\df(i+1)$ in $\pi$. 
Thus, $\sem{\pi,\theta}\le \frac12$.
\end{itemize}

\end{proof}

\subsection{Adding Unary Multiplication Operators}
\label{subsec: adding unary}
As we have seen in Section~\ref{sec:average}, adding the operator $\avg{}$ to $\DLTL$ makes model checking undecidable. One may still want to find propositional quality operators that we can add to the logic preserving its decidability. In this section we describe one such operator. We extend $\DLTL$ with the operator $\factorU_\lambda$, for $\lambda\in (0,1)$, with the semantics $\sem{\pi,\factorU_\lambda \phi}=\lambda\cdot \sem{\pi,\phi}$. This operator allows the specifier to manually change the satisfaction value of certain subformulas. This can be used to express importance, reliability, etc. of subformulas. For example, in $\Alw({\it  request} \rightarrow ({\it response} \vee \factorU_{\frac{2}{3}}\Next {\it response})$, we limit the satisfaction value of computations in which a response is given with a delay to $\frac{2}{3}$.

Note that the operator $\factorU_\lambda$ is similar to a one-time application of
$\Until_{\dfl{\lambda}^{+1}}$, thus $\factorU_\lambda \phi$ is equivalent to $\False \Until_{\dfl{\lambda}^{+1}} \psi$. In practice, it is better to handle $\factorU_\lambda$ formulas directly, by adding the following transitions to the construction in the proof of Theorem~\ref{thm:DLTL to AWW}.

\hspace{-.6cm}$\delta(\factorU_\lambda \phi>t,\sigma)=\begin{cases}
\delta(\phi>\frac{t}{\lambda},\sigma) & \text{if } \frac{t}{\lambda}<1,\\
\False & \text{if } \frac{t}{\lambda}\ge 1,
\end{cases}$%\hspace{.1cm}
$\delta(\factorU_\lambda \phi<t,\sigma)=\begin{cases}
\delta(\phi<\frac{t}{\lambda},\sigma) & \text{if } \frac{t}{\lambda}\le 1,\\
\True & \text{if } \frac{t}{\lambda}> 1.
\end{cases}$

\vspace*{-3mm}
\section{Extensions}% and future work}
\label{sec:ext}

\subsection{$\DLTL$ with Past Operators}
\label{ext:past}
One of the well-known augmentations of $\LTL$ is the addition of {\em past operators\/} \cite{LPZ85}. These operators enable the specification of exponentially more succinct formulas, while preserving the PSPACE complexity of model checking.
In this section, we add {\em discounting-past} operators to $\DLTL$, and show how to perform model-checking on the obtained logic.

We add the operators $\Yest \phi$, $\phi \Since \psi$, and $\phi \Since_\df \phi$ (for $\df\in \D$) to $\DLTL$, and denote the extended logic $\DPLTL$, with the following semantics. For $\DPLTL$ formulas $\phi, \psi$, a function $\df\in \D$, a computation $\pi$, and an index $i\in \Nat$, we have
\begin{itemize}
%orna3 no need to break lines
\item $\sem{\pi^i,\Yest \phi}=\sem{\pi^{i-1},\phi}$ if $i>0$, and $0$ otherwise.
\item $\sem{\pi^i,\phi \Since \psi}=\max_{0\le j\le i}\set{\min\set{\sem{\pi^j,\psi},\min_{j<k\le i}\set{\sem{\pi^k,\phi}}}}$.
\item $\sem{\pi^i,\phi \Since_\df \psi}=
\max_{0\le j\le i}\set{\min\set{\df(i-j)\sem{\pi^j,\psi},\min_{j<k\le i}\set{\df(i-k)\sem{\pi^k,\phi}}}}$.
%\item $\sem{\pi^i,\phi \Since \psi}=$\\$\max_{0\le j\le i}\set{\min\set{\sem{\pi^j,\psi},\min_{j<k\le i}\set{\sem{\pi^k,\phi}}}}$.
%\item $\sem{\pi^i,\phi \Since_\df \psi}=$\\
%\spb $\max_{0\le j\le i}\set{\min\set{\df(i-j)\sem{\pi^j,\psi},\min_{j<k\le i}\set{\df(i-k)\sem{\pi^k,\phi}}}}$.
\end{itemize}
Observe that since the past is finite, then the semantics for past operators can use $\min$ and $\max$ instead of $\inf$ and $\sup$. 

%shaullT14.1
As in $\DLTL$, our solution for the $\DPLTL$ model-checking problem is by translating $\DPLTL$ formulas to automata. The construction extends the construction for the Boolean case, which uses 2-way weak alternating automata ($\twAWW$). The details of the construction appear in Appendix~\ref{apx:LTL with past}. The use of the obtained automata in decision procedures is similar to that in Section~\ref{sec:alg proc}. In particular, it follows that the model-checking problem for $\DPLTL$ with exponential discounting is in PSPACE.

\subsection{Weighted Systems}
\label{wsys}
A central property of the logic $\LTLD$ is that the verified system need not be weighted in order to get a quantitative satisfaction -- it stems from taking into account the delays in satisfying the requirements. Nevertheless, $\LTLD$  also naturally fits weighted systems, where the atomic propositions have a value between $0$ and $1$.

A {\em weighted Kripke structure} is a tuple $\K=\zug{AP,S,I,\rho,L}$, where $AP, S, I$, and $\rho$ are as in Boolean Kripke structures, and $L:S\to [0,1]^{AP}$ maps each state to a weighted assignment to the atomic propositions. Thus, the value $L(s)(p)$ of an atomic proposition $p \in AP$ in a state $s \in S$ is a value in $[0,1]$. The semantics of $\DLTL$ with respect to a weighted computation coincides with the one for non-weighted systems, except that for an atomic proposition $p$, we have that $\sem{\pi, p}=L(\pi_0)(p)$.

It is possible to extend the construction of $\A_{\phi,v}$ described in Section~\ref{sec:DltlToAww} to an alphabet $W^{AP}$, where $W$ is a set of possible values for the atomic propositions. Indeed, we only have to adjust the transition for states that correspond to atomic propositions, as follows: for $p\in AP$, $v\in [0,1]$, and $\sigma\in W^{AP}$, we have that

\noindent {\labelitemi}~$\delta(p>v,\sigma)=\begin{cases}
\True & \text{ if } \sigma(p)>v,\\
\False & \text{ otherwise.}
\end{cases}$
\hspace{.6cm}
\noindent {\labelitemi}~$\delta(p<v,\sigma)=\begin{cases}
\True & \text{ if }\sigma(p)<v,\\
\False & \text{ otherwise.}
\end{cases}$

\subsection{Changing the Tendency of Discounting}\label{sec:DiscountingTendency}
One may observe that in our discounting scheme, the value of future formulas is discounted toward $0$. This, in a way, reflects an intuition that we are pessimistic about the future, or at least we are impatient. While in some cases this fits the needs of the specifier, it may well be the case that we are ambivalent to the future. To capture this notion, one may want the discounting to tend to $\frac12$. Other values are also possible. For example, it may be that we are optimistic about the future, say when a system improves its performance while running and we know that components are likely to function better in the future. We may then want the discounting to tend, say, to $\frac34$. 

%shaull2 - removed tilde-O.
To capture this notion, we define the operator $\parUntil_{\df,z}$, parameterized by 
%and $\dparUntil_{\df,z}$ 
$\df\in \D$ and $z\in [0,1]$, with the following semantics.
$\sem{\pi, \varphi \parUntil_{\df,z} \psi}  =
$\\ $
\sup\limits_{i\ge 0} \{ \min \{\df(i)\sem{\pi^i,\psi}+(1-\df(i))z,  \min\limits_{0\leq j < i} \df(j)\sem{\pi^j,\varphi}+(1-\df(j))z\}\}.$
% $$\sem{\pi, \varphi \dparUntil_{\df,z} \psi}  =\inf\limits_{i\ge 0} \{ \max \{\df(i)\sem{\pi^i,\psi}+(1-\df(i))z,  \max\limits_{0\leq j < i} \df(j)\sem{\pi^j,\varphi}+(1-\df(j))z\}\}.$$
The discounting function $\df$ determines the rate of convergence, and $z$ determines the limit of the discounting. The longer it takes to fulfill the ``eventuality'', the closer the satisfaction value gets to $z$.
We observe that $\phi\Until_{\df}\psi\equiv \phi\parUntil_{\df,0}\psi$.
\begin{example}
Consider a process scheduler. The scheduler decides which process to run at any given time. The scheduler may also run a defragment tool, but only if it is not in expense of other processes. This can be captured by the formula $\phi=\True \parUntil_{\eta,\frac12} {\rm defrag}$. Thus, the defragment tool is a ``bonus'': if it runs, then the satisfaction value is above $\frac12$, but if it does not run, the satisfaction value is $\frac12$. Treating $1$ as ``good'' and $0$ as ``bad'' means that $\frac12$ is ambivalent.
\end{example}

%One may verify that
%$\neg(\phi\parUntil_{\df,z}\psi)\equiv (\neg \phi)\dparUntil_{\df,1-z} (\neg \psi)$, so it is enough to include $\parUntil_{\df,z}$ in the logic.
%Also, clearly $\phi\Until_{\df}\psi\equiv \phi\parUntil_{\df,0}\psi$ and $\phi\dUntil_{\df}\psi=\phi\dparUntil_{\eta,1}\psi$.

We claim that Theorem~\ref{thm:DLTL to AWW} holds under the extension of $\DLTL$ with the operator $\parUntil$.
Indeed, the construction of the $\AWW$ is augmented as follows. For $\phi=\psi_1\parUntil_{\df,z}\psi_2$, denote $\frac{t-(1-\df(0))z}{\df(0)}$ by $\tau$. One may observe that the conditions on $\frac{t}{\df(0)}$ correspond to conditions on $\tau$ when dealing with $\parUntil$. Accordingly, the transitions from the state $(\psi_1\parUntil_{\df,z} \psi_2>t)$ are defined as follows.

First, if $t=z$, then $\tau=t$ and we identify the state $(\psi_1\parUntil_{\df,z} \psi_2>t)$ with the state
$(\psi_1\Until \psi_2>t)$. Otherwise, $z\neq t$ and we define:
\begin{itemize}
\item
$\delta((\psi_1\parUntil_{\df,z} \psi_2>t),\sigma)=$
 $\left[\begin{array}{ll}
\delta((\psi_2> \tau),\sigma)\vee & \\

\ \ [\delta((\psi_1>\tau),\sigma)\wedge (\psi_1\parUntil_{\df^{+1},z}\psi_2> t)] & \mbox{if $0\le \tau< 1$},\\
\False &  \mbox{if $\tau \geq 1$},\\
\True &  \mbox{if $\tau < 0$}.
\end{array}
\right.$

\item
$\delta((\psi_1\parUntil_{\df,z} \psi_2<t),\sigma)=
$
$
\left[\begin{array}{ll}
\delta((\psi_2< \tau),\sigma)\wedge & \\

\ \ [\delta((\psi_1<\tau),\sigma)\vee (\psi_1\parUntil_{\df^{+1},z}\psi_2< t)] & \mbox{if $0< \tau\le 1$},\\
\True &  \mbox{if $\tau > 1$},\\
\False &  \mbox{if $\tau \le 0$}.\\
\end{array}
\right.$
\end{itemize}
The correctness of the construction is proved in Appendix~\ref{apx:disc tendency proof}.
\section{Discussion}

An ability to specify and to reason about quality would take formal methods a significant step forward. 
%Beyond much more informative verification procedures, designers would be willing to give up manual design only when automatically synthesized methods would return systems of comparable quality.
Quality has many aspects, some of which are propositional, such as prioritizing one satisfaction scheme on top of another, and some are temporal, for example having higher quality for implementations with shorter delays. In this work we provided a solution for specifying and reasoning about temporal quality, augmenting the commonly used linear temporal logic (LTL). A satisfaction scheme, such as ours, that is based on elapsed times introduces a big challenge, as it implies infinitely many satisfaction values. Nonetheless, we showed the decidability of the model-checking problem, and for the natural exponential-decaying satisfactions, the complexity remains as the one for standard LTL, suggesting the interesting potential of the new scheme. As for combining propositional and temporal quality operators, we showed that the problem is, in general, undecidable, while certain combinations, such as adding priorities, preserve the decidability and the complexity.

\gap\noindent
{\bf Acknowledgement.} We thank Eleni Mandrali for pointing to an error in an earlier version of the paper.
\tiny
\bibliography{../ok}
\normalsize

\newpage
\appendix
\section{Proofs}

\subsection{Proof of Lemma~\ref{lem:LTL for positive}}

We construct $\pos{\phi}$ and $\notone{\phi}$ by induction on the structure of $\phi$ as follows. In all cases but the $\Until$ case we do not use the assumption that $\pi=u\cdot v^\omega$ and prove an ``iff'' criterion. 
\begin{itemize}
\item If $\phi$ is of the form $\True,\False$, or $p$, for an atomic proposition $p$, then $\pos{\phi}=\phi$ and $\notone{\phi}=\neg \phi$. Correctness is trivial.

\item If $\phi$ is of the form $\psi_1\vee \psi_2$, then $\pos{\phi}=\pos{\psi_1}\vee \pos{\psi_2}$  and $\notone{\phi}=\notone{\psi_1}\wedge \notone{\psi_2}$. Indeed, for every computation $\pi$ we have that $\sem{\pi,\phi}>0$ iff either $\sem{\pi,\psi_1}>0$ or $\sem{\pi,\psi_2}>0$, and  $\sem{\pi,\phi}<1$ iff both $\sem{\pi,\psi_1}<1$ and $\sem{\pi,\psi_2}<1$.

\item If $\phi$ is of the form $\Next \psi_1$, then $\pos{\phi}=\Next (\pos{\psi_1})$ and $\notone{\phi}=\Next (\notone{\psi_1})$. Correctness is trivial.

\item If $\phi$ is of the form $\psi_1\Until \psi_2$, then $\pos{\phi}=\pos{\psi_1}\Until \pos{\psi_2}$ and $\notone{\phi}=\neg((\neg (\notone{\psi_1}))\Until (\neg (\notone{\psi_2})))$. 

We start with $\pos{\phi}$. For every computation $\pi$ we have that $\sem{\pi,\phi}>0$ iff there exists $i\ge 0$ such that $\sem{\pi^i,\psi_2}>0$ and for every $0\le j<i$ it holds that $\sem{\pi^j,\psi_1}>0$. This happens iff $\pi$ satisfies $\pos{\psi_1}\Until \pos{\psi_2}$. 

Before we turn to the case of $\notone{\phi}$, let us note that
readers familiar with the release ($\Release$) operator of $\LTL$ may find it clearer to observe that $\notone{\phi}=\notone{\psi_1}\Release \notone{\psi_2}$, which perhaps gives a clearer intuition for the correctness of the construction.

Now, if $\sem{\pi,\phi}<1$, then for every $i\ge 0$ it holds that either $\sem{\pi^i,\psi_2}<1$ or $\sem{\pi^j,\psi_1}<1$ for some $0\le j<i$. Thus, for every $i\ge 0$, either $\pi^i\models \notone{\psi_2}$, or $\pi^j\models \notone{\psi_1}$ for some $0\le j<i$. So for every $i\ge 0$, either $\pi^i\not \models \neg (\notone{\psi_2})$, or $\pi^j\not\models \neg( \notone{\psi_1})$ for some $0\le j<i$. It follows that $\pi\not \models (\neg (\notone{\psi_1}))\Until (\neg (\notone{\psi_2}))$. Equivalently, $\pi\models \neg((\neg (\notone{\psi_1}))\Until (\neg (\notone{\psi_2})))$.

Conversely, if $\pi=u\cdot v^\omega$ and $\pi\models \neg((\neg (\notone{\psi_1}))\Until (\neg (\notone{\psi_2})))$, then $\pi\not \models (\neg (\notone{\psi_1}))\Until (\neg (\notone{\psi_2}))$, so for every $i\ge 0$, either $\pi^i\models \notone{\psi_2}$ or $\pi^j\models \notone{\psi_1}$ for some $0\le j< i$. By the induction hypothesis, for every $i\ge 0$, either $\sem{\pi^i,\psi_2}<1$ or $\sem{\pi^j,\psi_1}<1$ for some $0\le j< i$. We now use the assumption that $\pi=u\cdot v^\omega$ to observe that the $\sup$ in the expression for $\sem{\pi,\phi}$ is attained as a $\max$, as there are only finitely many distinct suffixes for $\pi$ (namely $\pi^0,...,\pi^{|u|+|v|-1}$). Thus, since all the elements in the $\max$ are strictly smaller than $1$, we conclude that $\sem{\pi,\phi}<1$.

\item  If $\phi=\neg \psi$, then $\pos{\phi}=\notone{\psi}$ and $\notone{\phi}=\pos{\psi}$. Again, correctness is trivial.

\item
If $\phi=\psi_1\Until_\df \psi_2$ for $\df \in \D$, then
$\pos{\phi}=\pos{\psi_1}\Until \pos{\psi_2}$. Indeed, since $\df(i) > 0$ for all $i \geq 0$, then 
$\pos{\phi}=\pos{\psi_1\Until \psi_2}$.

Now, $\notone{\phi}$ is defined as follows. First, if $\df(0)<1$, then $\notone{\phi}=\True$. If $\df(0)=1$, then $\notone{\phi}=\notone{\psi_2}$. Indeed, since $\df$ is strictly decreasing, the only chance of $\phi$ to have $\sem{\pi,\phi}=1$ is when both $\df(0)=1$ and $\sem{\pi^0,\psi_2}=1$. Since a satisfaction value cannot exceed $1$, the latter happens iff $\df(0)=1$ and $\pi\not\models \notone{\psi_2}$ (where the ``only if'' direction is valid when $\pi$ is a lasso, as is assumed). 

%\item
%%orna2: we don't need it, but we do need to eventually refer to the O_{\eta,z} case.
%If $\phi=\psi_1\dUntil_\df \psi_2$ for $\df \in \D$, then $\pos{\phi}= \pos{\psi_2}\Until (\pos{\psi_1}\wedge \pos{\psi_2})$. Clearly if $\pi\models\pos{\psi_2}\Until (\pos{\psi_1}\wedge \pos{\psi_2})$ then $\sem{\pi,\phi}>0$. Conversely, if $\sem{\pi,\phi}>0$, then there exists an index at which $\pos{\psi_1}\wedge \pos{\psi_2}$ is satisfied in $\pi$. Consider the minimal index $i$ at which this holds. It is easy to show by induction on $i$ that for every $j\le i$, $\psi_2$ must hold in $\pi^j$.
%
%Now, $\notone{\phi}$ is defined as follows. First, if $\df(0)<1$, then $\notone{\phi}=\True$. If $\df(0)=1$, then $\notone{\phi}=\notone{\psi_2}\vee \notone{\psi_1}$. Indeed, since $\df$ is strictly decreasing, then $\sem{\pi,\phi}=1$ iff $\df(0) =1$, $\sem{\pi^0,\psi_1}=1$ and $\sem{\pi^0,\psi_2}=1$. Since a satisfaction value cannot exceed $1$, the latter happens iff $\df(0)=1$ and both $\pi\not\models \notone{\psi_1}$ and $\pi\not\models \notone{\psi_2}$. Accordingly, when $\df(0)=1$, then $\sem{\pi,\phi}<1$ iff either $\sem{\pi,\psi_1}<1$ or $\sem{\pi,\psi_2}<1$. Again, the ``only if'' direction utilizes the assumption that $\pi$ is a lasso.
\end{itemize}

%shaull2
Finally, it is easy to see that $|\pos{\phi}|$ and $|\notone{\phi}|$ are both $O(|\phi|)$.
%\end{proof}

\subsection{The standard translation of $\LTL$ to $\AWW$}
\label{apx: standard LTL to AWW}
For completeness, we bring here the construction of the translation from $\LTL$ to $\AWW$, which we use in Theorem~\ref{thm:DLTL to AWW}. For the correctness proof, see e.g.~\cite{Var96}.

Given an \LTL formula $\phi$ over the atomic propositions $AP$, we construct an \AWW $\A_\phi=\zug{Q,2^{AP},Q_0,\delta,\alpha}$ as follows. The state space $Q$ consists of all the subformulas of $\phi$, and their negations (we identify $\neg\neg \psi$ with $\neg \psi$). The initial state is $\phi$, and the accepting states are all the formulas of the form $\neg(\psi_1\Until\psi_2)$. It remains to define the transition function.

We start with a few notations. For a Boolean formula $\theta$ over $Q$, we define its {\em dual formula} $\overline{\theta}$ by induction over the construction of $\theta$, as follows.
\begin{itemize}
\item For $\psi\in Q$ we have $\overline{\psi}=\neg \psi$.
\item $\overline{\True}=\False$ and $\overline{\False}=\True$.
\item $\overline{\psi_1\vee\psi_2}=\overline{\psi_1}\wedge \overline{\psi_2}$ and $\overline{\psi_1\wedge\psi_2}=\overline{\psi_1}\vee \overline{\psi_2}$
\end{itemize}
The transition function can now be defined as follows. Let $\psi\in Q$ and $\sigma\in 2^{AP}$.
\begin{itemize}
\item If $\psi=p\in AP$, then $\delta(\psi,\sigma)=\begin{cases}
\True & p\in \sigma,\\
\False & p\notin \sigma,
\end{cases}$
\item If $\psi=\zeta _1\vee\zeta_2$, then $\delta(\psi,\sigma)=\delta(\zeta_1,\sigma)\vee\delta(\zeta_2,\sigma)$,
\item If $\psi=\neg \zeta$, then $\delta(\psi,\sigma)=\overline{\delta(\zeta,\sigma)}$,
\item If $\psi=\Next \zeta$, then $\delta(\psi,\sigma)=\zeta$,
\item If $\psi=\zeta_1\Until\zeta_2$, then $\delta(\psi,\sigma)=\delta(\zeta_2,\sigma)\vee(\delta(\zeta_1,\sigma)\wedge \zeta_1\Until\zeta_2)$.
\end{itemize}

\subsection{Continuation of the Proof of Theorem~\ref{thm:DLTL to AWW}}
%\begin{proof}
We continue the proof that is given in the main text, showing that the constructed $\AWW$ $\A_{\phi,v}$ is indeed finite and correct.

We first make some notations and observations regarding the structure of $\A_{\phi,v}$. For every state $(\psi > t)$ (resp. $(\psi<t)$) we refer to $\psi$ and $t$ as the state's {\em formula} and {\em threshold}, respectively. If the outermost operator in $\psi$ is a discounting operator, then we refer to its discounting function as the state's {\em discounting function}. For states of Type-2 we refer to their {\em formula} only (as there is no threshold).

First observe that the only cycles in $\A_{\phi,v}$ are self-loops. Indeed, consider a transition from state $q$ to state $s\neq q$. Let $\psi_q,\psi_s$ be the formulas of $q$ and $s$, respectively.
Going over the different transitions, one may see that either $\psi_s$ is a strict subformula of $\psi_q$, or $s$ is a Type-2 state, or both $\psi_q$ and $\psi_s$ have an outermost discounting operator with discounting functions $\df$ and $\df^{+1}$ respectively. By induction over the construction of $\phi$, this observation proves that there are only self-cycles in $\A_{\phi_v}$.

We now observe that in every run of $\A_{\phi,v}$ on an infinite word $w$, every infinite branch (i.e., a branch that does not reach $\True$) must eventually be in a state of the form $\psi_1\Until\psi_2>t$, $\psi_1\Until\psi_2<t$, $\psi_1\Until\psi_2$ or $\neg(\psi_1\Until\psi_2)$ (if it's a Type-2 state). Indeed, these states are the only states that have a self-loop, and the only cycles in the automaton are self-loops.

We start by proving that there are finitely many states in the construction. First, all the sub-automata that correspond to Type-2 states have $O(|\varphi|)$ states. This follows immediately from Lemma~\ref{lem:LTL for positive} and from the construction of an $\AWW$ from an $\LTL$ formula.

Next, observe that the number of possible state-formulas, up to differences in the discounting function, is $O(|\phi|)$. Indeed, this is simply the standard closure of $\phi$. It remains to prove that the number of possible thresholds and discounting functions is finite.

We start by claiming that for every threshold $t>0$, there are only finitely many reachable states with threshold $t$. Indeed, for every discounting function $\df\in \D$ (that appears in $\phi$), let $i_{t,\df}=\max\set{i: \frac{t}{\df(i)}\le 1}$. The value of $i_{t,\df}$ is defined, since the functions tend to $0$. Observe that in every transition from a state with threshold $t$, if the next state is also with threshold $t$, then the discounting function (if relevant) is either some $\df'\in \D$, or $\df^{+1}$. There are only finitely many functions of the former kind. As for the latter kind, after taking $\df^{+1}$ $i_{t,\df}$ times, we have that $t/\df^{+i_{t,\df}}(0)>1$. By the definition of $\delta$, in this case the transitions are to the Boolean-formula states (i.e., $\True,\False$, or some $\pos{\psi}$), from which there are finitely many reachable states. We conclude that for every threshold, there are only finitely many reachable states with this threshold.

Next, we claim that there are only finitely many reachable thresholds. This follows immediately from the claim above. We start from the state $\phi>v$. From this state, there are only finitely many reachable discounting functions. The next threshold that can be encountered is either $1-v$, or $\frac{v}{\df(0)}$ for $\df$ that is either in $\D$ or one of the $\df^{+i}$ for $i\le i_{v,\df}$. Thus, there are only finitely many such thresholds. Further observe that if a different threshold is encountered, then by the definition of $\delta$, the state's formula is deeper in the generating tree of $\phi$. Thus, there are only finitely many times that a threshold can change along a single path. So by induction over the depth of the generating tree, we can conclude that there are only finitely many reachable thresholds.

We conclude that the number of states of the automaton is finite.

Next, we prove the correctness of the construction. From Lemma~\ref{lem:LTL for positive} and the correctness of the standard translation of $\LTL$ to $\AWW$, it remains to prove that for every path $\pi$ and for every state $(\psi > v)$ (resp. $(\psi < v)$):
\begin{enumerate}
\item If $\sem{\pi,\psi}>v$ (resp. $\sem{\pi,\psi}<v$), then $\pi$ is accepted from $(\psi > v)$ (resp. $(\psi < v)$).
\item If $\pi=u\cdot v^\omega$ and $\pi$ is accepted from state $(\psi > v)$ (resp. $(\psi < v)$) then $\sem{\pi,\psi}>v$ (resp. $\sem{\pi,\psi}<v$).
\end{enumerate} 

The proof is by induction over the construction of $\phi$, and is fairly trivial given the definition of $\delta$.

%\end{proof}

\subsection{Proof of Lemma~\ref{lem:DLTL to NBW}}
%\begin{proof}
We start by defining {\em generalized B\"uchi automata}. An $\NGBW$ is $\A=\zug{Q,\Sigma,\delta,Q_0,\alpha}$, where $Q,\Sigma,\delta,Q_0$ are as in $\NBW$. The acceptance condition is $\alpha=\set{F_1,...,F_k}$ where $F_i\subseteq Q$ for every $1\le i\le k$. A run $r$ of $\A$ is accepting if for every $1\le i\le k$, $r$ visits $F_i$ infinitely often.

We now proceed with the proof.

Consider the $\AWW$ $\A$ obtained from $\phi$ using the construction of Section~\ref{sec:DltlToAww}.

In the translations of $\AWW$ to $\NBW$ using the method of~\cite{GO01}, the $\AWW$ is translated to an $\NGBW$ whose states are the subset-construction of the $\AWW$.
This gives an exponential blowup in the size of the automaton. We claim that in our translation, we can, in a sense, avoid this blowup.

Intuitively, each state in the $\NGBW$ corresponds to a conjunction of states of the $\AWW$. Consider such a conjunction of states of $\A$. If the conjunction contains two states $(\psi<t_1)$ and $(\psi < t_2)$, and we have that $t_1<t_2$, then by the correctness proof of Theorem~\ref{thm:DLTL to AWW}, it holds that a path $\pi$ is accepted from both states, iff $\pi$ is accepted from $(\psi<t_1)$. Thus, in every conjunction of states from $\A$, there is never a need to consider a formula with two different ``$<$'' thresholds. Dually, every formula can appear with at most one ``$>$'' threshold.

Next, consider conjunctions that contain states of the form $(\psi_1\Until_\df \psi_2<t)$ and $(\psi_1\Until_{\df^{+k}} \psi_2<t)$. Again, since the former assertion implies the latter, there is never a need to consider two such formulas. Similar observations hold for the other discounting operators.

Thus, we can restrict the construction of the $\NGBW$ to states that are conjunctions of states from the $\AWW$, such that no discounting operator appears with two different ``offsets''.

Further observe that by the construction of the $\AWW$, the threshold of a discounting formula does not change, with the transition to the same discounting formula, only the offset changes. That is, from the state $(\psi_1\Until_\df \psi_2<t)$, every reachable state whose formula is $\psi_1\Until_{\df^{+k}} \psi_2$ has threshold $t$ as well. Accordingly, the possible number of thresholds that can appear with the formula $\psi_1\Until_\df \psi_2$ in the subset construction of $\A$, is the number of times that this formula appears as a subformula of $\phi$, which is $O(|\phi|)$.

We conclude that each state of the obtained $\NGBW$ is a function that assigns each subformula\footnote{where a subformula may have several occurences, e.g., in the formula $p\wedge \Next p$ we have two occurences of the subformula $p$} of $\phi$ two thresholds.
The number of possible thresholds and offsets is linear in the number of states of $\A$, thus, the number of states of the $\NGBW$ is $|\A|^{O(|\phi|)}$.

Finally, translating the $\NGBW$ to an $\NBW$ requires multiplying the size of the state space by $|\A|$, so the size of the obtained $\NBW$ is also $|\A|^{O(|\phi|)}$.
%\end{proof}

\subsection{Proof of Theorem~\ref{thm: exp disc to AWW}}
%\begin{proof}
We construct an \AWW $\A=\A_{\phi,v}$ as per Section~\ref{sec:DltlToAww}, with some changes.

Recall that the ``interesting'' states in $\A$ are those of the form $\psi_1\Until_{\df^{+i}}\psi_2$.
%(resp. $\psi_1\dUntil_{\df^{+i}}\psi_2$). 
Observe that for the function $\dfl{\lambda}$ it holds that $\dfl{\lambda}^{+i}=\lambda^i\cdot \dfl{\lambda}$. Accordingly, we can replace a state of the form $\psi_1\Until_{\dfl{\lambda}^{+i}}\psi_2<t$ with the state $\psi_1\Until_{\dfl{\lambda}}
\psi_2<\frac{t}{\lambda^i}$, as they express the same assertion.
%(and similarly for $\dUntil$, and for $>t$). 
Finally, notice that $\dfl{\lambda}(0)=1$. Thus, we can simplify the construction of $\A$ with the following transitions:

Let $\sigma\in 2^{AP}$, then we have that
\begin{itemize}
%%%%%%Until
\item
$\delta((\psi_1\Until_{\dfl{\lambda}}
\psi_2>t),\sigma)=$\\$\left[ \begin{array}{ll}
\delta((\psi_2> t),\sigma)\vee & \\

\ \ [\delta((\psi_1>t),\sigma)\wedge (\psi_1\Until_{\dfl{\lambda}}
\psi_2> \frac{t}{\lambda})] & \mbox{ if $0<t< 1$},\\
\False &  \mbox{ if $t \geq 1$},\\
\delta((\pos{(\psi_1\Until_{\dfl{\lambda}}
\psi_2)}),\sigma) & \mbox{ if $t=0$}.
\end{array}
\right.
$

\item
$\delta((\psi_1\Until_{\dfl{\lambda}}\psi_2<t),\sigma)=$\\$\left[ \begin{array}{ll}
\delta((\psi_2<t),\sigma)\wedge & \\

\ \ [\delta((\psi_1<t),\sigma)\vee (\psi_1\Until_{\dfl{\lambda}}
\psi_2<\frac{t}{\lambda})] & \mbox{ if $0<t\le 1$},\\
\True &  \mbox{ if $t> 1$},\\
\False & \mbox{ if $t=0$ }.
\end{array}
\right.
$

%%%%%%dUntil
\item
$\delta((\psi_1\Until_{\dfl{\lambda}}
\psi_2>t),\sigma)=$\\$\left[ \begin{array}{ll}
	\delta((\psi_2>t),\sigma)\wedge & \\

\ \ [\delta((\psi_1>t),\sigma)\vee (\psi_1\Until_{\dfl{\lambda}}
\psi_2>\frac{t}{\lambda})] & \mbox{ if $0<t< 1$},\\
\False &  \mbox{ if $t \geq 1$},\\
\delta((\pos{(\psi_1\Until_{\dfl{\lambda}}
\psi_2)}),\sigma) & \mbox{ if $t=0$}.
\end{array}
\right.
$

\item
$\delta((\psi_1\Until_{\dfl{\lambda}}\psi_2<t),\sigma)=$\\$\left[ \begin{array}{ll}
\delta((\psi_2<t),\sigma)\vee & \\

\ \ [\delta((\psi_1<t),\sigma)\wedge (\psi_1\Until_{\dfl{\lambda}}
\psi_2<\frac{t}{\lambda})] & \mbox{ if $0<t\le 1$},\\
\True &  \mbox{ if $t> 1$},\\
\False & \mbox{ if $t=0$}.
\end{array}
\right.
$
\end{itemize}

The correctness and finiteness of the construction follows from Theorem~\ref{thm:DLTL to AWW}, with the observation above. We now turn to analyze the number of states in $\A$.

For every state $(\psi > t)$ (resp. $(\psi<t)$) we refer to $\psi$ and $t$ as the state's {\em formula} and {\em threshold}, respectively. Observe that the number of possible state-formulas is $O(|\phi|)$.
Indeed, the formulas in the states are either in the closure of $\phi$, or are of the form $\pos{\psi}$, where $\psi$ is in the closure of $\phi$. This is because in the new transitions we do not carry the offset, but rather change the threshold, so the state formula does not change.

It remains to bound the number of possible thresholds.
Consider a state with threshold $t$. In every succeeding state\footnote{This is almost correct. In fact, since $\delta$ is defined inductively, we may go through several transitions.}, the threshold (if exists) can either remain $t$, or change to $1-t$ (in case of negation), or $t/\lambda$, where $\lambda\in \fac(\phi)$, providing $t <1$.

Initially, we ignore negations. In this case, the number of states that can be reached from a threshold $t$ is bounded by the size of the set
$$\set{\prod_{i=1}^k \lambda_i: \prod_{i=1}^k \lambda_i>t,\ k\in \Nat, \forall i\ \lambda_i\in \fac(\phi)}$$
This length of the products can be bounded by $\log_\mu t=\frac{\log t}{\log \mu}=O(\log t)$, where $\mu=\max\fac(\phi)$. Thus, the number of possible values is bounded by $(\log_\mu t) ^{|\fac(\phi)|}$.
From here we denote $\log_\mu t$ by $\ell$.

Next, we consider negations. Observe that in every path, there are at most $|\phi|$ negations. In every negation, if the current state has threshold $s$, the threshold changes to $1-s$. We already proved that from threshold $t$ we can get at most $(\ell) ^m$ states. Furthermore, every such state has a threshold of the form
$$\frac{t}{\prod_{i=1}^{\ell} \lambda_i}$$
where $\lambda_i\in\fac(\phi)\cup\set{1}$ (we allow 1 instead of allowing shorter products).

Consider a negation state with threshold $s=\frac{t}{\prod_{i=1}^{\ell} \lambda_i}$. If $s=1$ then in the next state the threshold is $0$, and every reachable state is part of a Boolean $\AWW$ corresponding to an $\LTL$ formula, and thus has polynomially many reachable states.

Otherwise, we have that $s<1$. Denote $\fac(\phi)=\set{\lambda_1,...,\lambda_m}$, and assume that $t,\lambda_1,...,\lambda_m$ can be written as $t=\frac{p_t}{q_t}$ and for all $i$, $\lambda_i=\frac{p_i}{q_i}$ (which is possible as they are in $\Rat$), then we have that

$\log(1-s)=\log(1-\frac{t}{\prod_{i=1}^{\ell} \lambda_i})=\log (\prod_{i=1}^{\ell}\lambda_i-t)-\log(\prod_{i=1}^{\ell}\lambda_i)$ $>\log (\prod_{i=1}^{\ell}\lambda_i-t)$,
where the last transition is because $\log(\prod_{i=1}^{\ell}\lambda_i)<0$.

Let $q_{max}=\max\set{q_1,...,q_m}$, we now observe that
\begin{align*}
&\prod_{i=1}^{\ell}\lambda_i-t=\prod_{i=1}^{\ell}\frac{p_i}{q_i}-\frac{p_t}{q_t}=
\frac{\prod_{i=1}^{\ell}p_iq_t- \prod_{i=1}^{\ell}q_ip_t}{\prod_{i=1}^{\ell} q_iq_t}\ge\\  
&\frac{1}{\prod_{i=1}^{\ell} q_iq_t}\ge \frac{1}{q_{max}^{\ell }q_t}
\end{align*}

where the last transitions are because all the numbers are natural, and $s>1$, so the numerator must be a positive integer.
From this we get that
$$\log(1-s)>\log\left( \frac{1}{q_{max}^{\ell }q_t}\right)=-(\ell\cdot (\log q_{max}+\log q_t))$$
We see that we can bound $1-s$ from below by a rational number whose representation is linear in that of $t$ and $|\zug{\phi}|$. Denote this linear function by $f(n)$.
Thus, continuing in this manner through $O(\phi)$ negations, starting from the threshold $v$, we have that the number of thresholds reachable from every state is bounded by
$(\underbrace{f(f(...f(\log v)))}_{O(\varphi)})^m$
which is single exponential in $|\zug{\varphi}|$ and the description of $v$.

We conclude that the number of states of the automaton is single-exponential.
%\end{proof}

\subsection{Proof of Theorem~\ref{thm:synthesis}}
\label{apx: synthesis}
Recall that if $\pi$ is a computation such that $\sem{\pi,\phi}>v$, hen $\A_{\phi,v}$ accepts $\pi$. The converse however, is not true. Still, by carefully examining the construction in Theorem~\ref{sec:DltlToAww}, we observe that if $\A_{\phi,v}$ accepts a computation $\pi$, then $\sem{\pi,\phi}\ge v$ (note the non-strict inequality). 

Assume a partition of the letters in $AP$ to input and output signals, denoted $I$ and $O$, respectively. By following standard (Boolean) procedures for synthesis (see \cite{PR89a}), we can generate from $\A_{\phi,v}$ a deterministic tree automaton $\D$ that accepts a $2^O$-labeled $2^I$-tree iff all the paths along the tree are accepted in $\A_{\phi,v}$. A tree that is accepted in this manner represents a transducer that realizes $\phi$ with value at least $v$. Accordingly, if there exists a transducer $\T$, all of whose computations satisfy $\sem{\pi,\phi}>v$, then the tree that represents such a transducer is accepted in $\D$. Thus, the language of $\D$ is non-empty, and a non-emptiness witness is a transducer (tree) $\T'$, all of whose computations 
%are accepted in $A_{\phi,v}$, which implies that every computation $\pi$ of $\T'$ satisfies
satisfy $\sem{\pi,\phi}\ge v$.

%transducer $T'$ all of whose computations $\pi$ satisfy $\sem{\pi,\phi}\ge v$. (Note that $T'$ differs from $T$, satisfying $\ge v$ rather than $> v$.) This follows by observing that, by the construction of $\A_{\phi,v}$, if it accepts $\pi$ then $\sem{\pi,\phi}\ge v$. 
%%(If $\pi$ is a lasso computation, we know that $\sem{\pi,\phi}>v$. Here we observe that otherwise, it is still the case that $\sem{\pi,\phi}\ge v$).

We remark that this solution is only partial, as there might be an $I/O$-transducer whose computations $\pi$ all satisfy $\sem{\pi,\phi}\ge v$, but we will not find it.

\subsection{Translating $\DPLTL$ formulas to $\twAWW$}
\label{apx:LTL with past}
As we now show, the construction we use when working with the ``infinite future'' $\Until_\df$ operator is similar to that of the one we use for the ``finite past'' $\Since_\df$ operator.
The key for this somewhat surprising similarity is the fact our construction is based on a threshold. Under this threshold, we essentially bound the future that needs to be considered, thus the fact that it is technically infinite plays no role.

%We now 
%wish to 
%translate $\DPLTL$ formulas to automata. In the Boolean case, one of the common approaches to this is to use 2-way weak alternating automata ($\twAWW$). 
%shaull4
%short 
%In the Boolean case, there are two approaches to this: one is to use 2-way weak alternating automata ($\twAWW$), and the other goes through nondeterministic automata. To comply with our use of $\AWW$ in Theorem~\ref{thm:DLTL to AWW}, we work with $\twAWW$. 

%As in Theorem~\ref{thm:DLTL to AWW}, the construction combines Boolean PLTL formulas for satisfaction values in $\set{0,1}$ (in fact, since past is finite, here their construction is simpler). 

%We describes the construction in details in the appendix. The analysis of the blowup, and the use in decision procedures are similar to these in Section~\ref{sec:alg proc}. In particular, it follows that the model-checking problem for $\DPLTL$ with exponential discounting is in PSPACE.

%short

A $\twAWW$ is a tuple $\A=\zug{\Sigma,Q,q_0,\delta,\alpha}$ where $\Sigma,Q,q_0,\alpha$ are as in $\AWW$. The transition function is $\delta:Q\times \Sigma\to \B^+(\set{-1,1}\times Q)$. That is, positive Boolean formulas over atoms of the form $\set{-1,1}\times Q$, describing both the state to which the automaton moves and the direction in which the reading head proceeds.
%Reading a $-1$ transition in the beginning of the word moves the automaton to $\False$. See~\cite{KPV01} for details.

%shaull4
%short 
As in $\DLTL$, the construction extends the construction for the Boolean case. 
%, so we need to extend Lemma~\ref{lem:LTL for positive} and generate Boolean PLTL formulas for satisfaction values in $\set{0,1}$. 
%It is not hard to see that defining $\pos{(\psi_1 \Since_\df \psi_2)}=(\pos{\psi_1})\Since (\pos{\psi_2})$ and $\pos{(\Yest \psi)}=\Yest (\pos{\psi})$ works.
It is not hard to extend Lemma~\ref{lem:LTL for positive} and generate Boolean PLTL formulas for satisfaction values in $\set{0,1}$. 

Given a $\DPLTL$ formula $\phi$ and a threshold $t\in [0,1)$, we construct a $\twAWW$ as in
Theorem~\ref{thm:DLTL to AWW} with the following additional transitions:\footnote{In addition, the atoms in the transitions in Theorem~\ref{thm:DLTL to AWW} are adjusted to the $\twAWW$ syntax by replacing each atom $q$ by the atom $\zug{1,q}$.}.
\begin{itemize}
\item $\delta((\Yest \psi > t),\sigma)=\zug{-1,(\psi>t)}$
\item $\delta((\Yest \psi < t),\sigma)=\zug{-1,(\psi<t)}$.
\item $\delta((\psi_1\Since \psi_2 >t),\sigma)=\delta((\psi_2>t),\sigma)\vee (\delta((\psi_1>t),\sigma) \wedge \zug{-1,(\psi_1\Since \psi_2>t)})$.
\item $\delta((\psi_1\Since \psi_2 <t),\sigma)=\delta((\psi_2<t),\sigma)\wedge (\delta((\psi_1<t),\sigma) \vee \zug{-1,(\psi_1\Since \psi_2<t)})$.

\item $\delta((\psi_1\Since_\df \psi_2 >t),\sigma)=$
\\ \spb
$\left[ \begin{array}{ll}
\delta((\psi_2>\frac{t}{\df(0)}),\sigma)\vee 
[\delta((\psi_1>\frac{t}{\df(0)}),\sigma) \wedge \zug{-1,(\psi_1\Since_{\df^{+1}} \psi_2>t)}] &  \mbox{ if $0<\frac{t}{\df(0)}< 1$},\\
\False & \mbox{ if $\frac{t}{\df(0)} \geq 1$},\\
\delta((\pos{(\psi_1\Since_\df \psi_2)}),\sigma) & \mbox{ if $\frac{t}{\df(0)}=0$}.
\end{array}
\right.$

\item $\delta((\psi_1\Since_\df \psi_2 <t),\sigma)=$
\\ \spb
$\left[ \begin{array}{ll}
\delta((\psi_2<\frac{t}{\df(0)}),\sigma)\wedge [\delta((\psi_1<\frac{t}{\df(0)}),\sigma) \vee \zug{-1,(\psi_1\Since_{\df^{+1}} \psi_2<t)}] &  \mbox{ if $0<\frac{t}{\df(0)}\le 1$},\\
\True & \mbox{ if $\frac{t}{\df(0)} > 1$},\\
\False & \mbox{ if $\frac{t}{\df(0)}=0$ }.
\end{array}
\right.$
\end{itemize}

%An operator $\dSince_\df$ (dual-since), with the semantics
%$\sem{\pi^i,\phi \dSince_\df \psi} = \min_{0\le j\le i} \{ \max$ $\{\df(i-j)$ $\sem{\pi^j,\psi},\max_{j<k\le i}\set{\df(i-k)\sem{\pi^k,\phi}}\}\}$,
%can be added in a similar manner.

The correctness of the construction and the analysis of the blowup
%, and the use in decision procedures
 are similar to those in Section~\ref{sec:alg proc}. 
 %In particular, it follows that the model-checking problem for $\DPLTL$ with exponential discounting is in PSPACE.

\subsection{Correctness Proof of the Construction in Section~\ref{sec:DiscountingTendency}}
\label{apx:disc tendency proof}
Consider the case of $(\psi_1\parUntil_{\df,z} \psi_2>t)$ (the dual case is similar). We check when this assertion holds.

First, if $z=t$, then this assertion is equivalent to $\psi_1\Until\psi_2>t$ (this follows directly from the semantics).
Thus, assume that $z\neq t$.

If $\tau<0$, then $t<(1-\df(0))z$, so in particular, for every value of $\sem{\pi^0,\psi_2}$, it holds that $t<\df(0)\sem{\pi^0,\psi_2}+(1-\df(0))z$, so the assertion is true, since the first operand in the $\sup$ is greater than $t$.

If $\tau\ge 1$ then $\df(0)+(1-\df(0))z\le t$, so both
$\df(0)\sem{\pi^0,\psi_2}+(1-\df(0))z\le t$ and $\df(0)\sem{\pi^0,\psi_1}+(1-\df(0))z\le t$. Thus, every operand in the $\sup$ has an element less than $t$, so the $\sup$ cannot be greater than $t$, so the assertion is false.

If $0\le \tau<1$, then similarly to the case of $\Until_\df$ - the assertion is equivalent to the following: either $\psi_2>\tau$, or both $\psi_1>\tau$ and $(\psi_1\parUntil_{\df^{+1},z} \psi_2>t)$.
%The case of $\dparUntil$ is proved similarly.

It remains to show that there are still only finitely many states. Since $z\neq t$, we get that $\lim_{\df(0)\to 0}\tau=\begin{cases}
\infty & t>z\\
-\infty & t<z
\end{cases}.$
Thus, after a certain number of transitions, $\tau$ takes a value that is not in $[0,1]$, in which case the next state is $\True$ or $\False$, so the number of states reachable from $\phi>t$ is finite.

We remark that this is not true if $z=t$, which is why we needed to treat this case separately.
\end{document}